\documentclass[sigconf,nonacm]{acmart}
\usepackage{balance}
\pdfoutput=1

\usepackage{amsfonts}
\usepackage{amsmath}
\usepackage{amssymb}
\usepackage{amsthm}
\usepackage{epsfig}
\usepackage[T1]{fontenc}
\usepackage{graphicx}
\usepackage[utf8]{inputenc}
\usepackage{url}
\usepackage{xspace}

\usepackage{float}
\usepackage{latexsym}
\usepackage[linesnumbered,ruled,vlined]{algorithm2e}
\usepackage{algpseudocode}
\usepackage{array}
\usepackage{diagbox}

\newtheorem{definition}{Definition}
\newtheorem{example}{Example}
\newtheorem{lemma}{Lemma}
\newtheorem{theorem}{Theorem}

\newcommand{\SDI}{{\sffamily\tt{SDI}}}

\newcommand{\SDIRS}{{\sffamily\tt{SDI-RS}}}

\newcommand{\RS}{{\sffamily\tt{RangeSearch}}}

\newcommand{\BFS}{{\sffamily\tt{BFS}}}
\newcommand{\DFS}{{\sffamily\tt{DFS}}}

\newcommand{\BBS}{{\sffamily\tt{BBS}}}
\newcommand{\BITMAP}{{\sffamily\tt{Bitmap}}}
\newcommand{\BNL}{{\sffamily\tt{BNL}}}
\newcommand{\INDEX}{{\sffamily\tt{Index}}}
\newcommand{\LESS}{{\sffamily\tt{LESS}}}
\newcommand{\NL}{{\sffamily\tt{NL}}}
\newcommand{\NN}{{\sffamily\tt{NN}}}
\newcommand{\SFS}{{\sffamily\tt{SFS}}}
\newcommand{\SALSA}{{\sffamily\tt{SaLSa}}}
\newcommand{\SUBSKY}{{\sffamily\tt{SUBSKY}}}
\newcommand{\ZINC}{{\sffamily\tt{ZINC}}}
\newcommand{\ZSEARCH}{{\sffamily\tt{ZSearch}}}

\usepackage{xcolor}

\balance

\begin{document}

\title{An Efficient Skyline Computation Framework}

\author{Rui Liu}
\email{rui.liu@univ-tours.fr}
\affiliation{%
 \institution{University of Tours}
 \country{France}
}

\author{Dominique Li}
\email{dominique.li@univ-tours.fr}
\affiliation{%
 \institution{ LIFAT Laboratory, University of Tours}
 \country{France}
}

\begin{abstract}
Skyline computation aims at looking for the set of tuples that are not worse than any other tuples in all dimensions from a multidimensional database.
In this paper, we present {\SDI} (Skyline on Dimension Index), a dimension indexing conducted general framework to skyline computation.
We prove that to determine whether a tuple belongs to the skyline, it is enough to compare this tuple with a bounded subset of skyline tuples in an arbitrary dimensional index, but not with all existing skyline tuples.
Base on {\SDI}, we also show that any skyline tuple can be used to stop the whole skyline computation process with outputting the complete set of all skyline tuples.
We develop an efficient algorithm {\SDIRS} that significantly reduces the skyline computation time, of which the space and time complexity can be guaranteed.
Our experimental evaluation shows that {\SDIRS} outperforms the baseline algorithms in general and is especially very efficient on high-dimensional data.

\end{abstract}

\maketitle

\section{Introduction}

Skyline computation aims at looking for the set of tuples that are not worse than any other tuples in all dimensions with respect to given criteria from a multidimensional database.
Indeed, the formal concept of {\em skyline} was first proposed in 2001 by extending SQL queries to find interesting tuples with respect to multiple criteria \cite{Borzsony2001Operator}, with the notion of {\em dominance}: we say that a tuple $t$ {\em dominates} another tuple $t^\prime$ if and only if for each dimension, the value in $t$ is better than the respective value in $t^\prime$.
The predicate {\em better} can be defined by any total order, such as {\em less than} or {\em greater than}.

\begin{figure}[htbp]
\begin{center}
\includegraphics[scale=0.6]{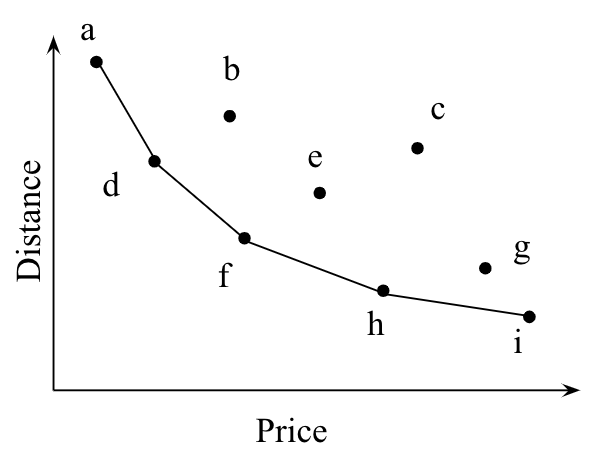}
\end{center}
\caption{The hotel Skyline on distance and price.}
\label{fig:hotel}
\end{figure}

For instance, Figure \ref{fig:hotel} shows the most mentioned example in skyline related literature where we consider the prices of hotels with respect to their distances to the city center (sometimes to the beach, or to the railway station, etc.).
If we are interested in hotels which are not only cheap but also close to the city center (the less value is the better), those represented by $a$, $d$, $f$, $h$, and $i$ constitute the skyline.
It's obvious that the hotel $d$ dominates the hotel $b$ since $d$ is better than $b$ in both distance and price; however, $a$ does not dominate $b$ because $a$ is better than $b$ in price but $b$ is however better than $a$ in distance.
In real-world database and user centric applications, such {\em Price-Distance} liked queries are doubtless interesting and useful, and have been widely recognized.

Since the first proposed {\BNL} algorithm \cite{Borzsony2001Operator}, the skyline computation problem has been deeply studied for about two decades and many algorithms have been developed to compute the Skyline, such as {\BITMAP}/{\INDEX} \cite{Tan2001IndexBitmap}, {\NN} \cite{Kossmann2002NN}, {\BBS} \cite{Papadias2005BBS}, {\SFS} \cite{Jan2005SFS}, {\LESS} \cite{Godfrey2005Less}, {\SALSA} \cite{Bartolini2006SaLSa}, {\SUBSKY} \cite{Tao2007Subsky}, {\ZSEARCH} \cite{Lee2007ZSearch}, and {\ZINC} \cite{Liu2010ZINC}.
However, the efficiencies brought by existing algorithms often depend on either complex data structures or specific application/data settings.
For instance, {\BITMAP} is based on a bitmap representation of dimensional values but is also limited by the cardinality; {\INDEX} is integrated into a B$^+$-tree construction process; {\NN} and {\BBS} rely to a specific data structure as R-tree besides {\NN} handles difficultly high-dimensional data (for instance $d > 4$ \cite{Tan2001IndexBitmap}); {\SUBSKY} specifically requires the B-tree and tuning the number of anchors; {\ZSEARCH} and {\ZINC} are built on top of the ZB-tree.
On the other hand, as well as paralleling Skyline computation \cite{chester2015scalable}, different variants of standard skylines have been defined and studied, such as top-k Skylines \cite{Tao2007Subsky}, streaming data Skylines \cite{lin2005stabbing}, partial ordered Skylines \cite{Liu2010ZINC}, etc., which are not in our scope.

In this paper, we present the {\SDI} ({\em Skyline on Dimension Index}) framework that allows efficient skyline computation by indexing dimensional values.
We first introduce the notion of {\em dimensional index}, based on which we prove that in order to determine whether a tuple belongs the skyline, it is enough to compare it only with existing skyline tuples present in any one dimensional index instead of comparing it with all existing skyline tuples, which can significantly reduce the total count of dominance comparisons while computing the skyline.
Furthermore, within the context of dimension indexing, we show that in most cases, one comparison instead of two is enough to confirm the dominance relation.
These properties can significantly reduce the total count of dominance comparisons that is mostly the bottleneck of skyline computation.
Different form all existing sorting/indexing based skyline algorithms, the application of dimension indexing allows to extend skyline computation to any total ordered categorical data such as, for instance, user preference on colors and on forms.
Based on {\SDI}, we also prove that any skyline tuple can be used to define a {\em stop line} crossing dimension indexes to terminate the skyline computation, which is in particular efficient on correlated data.
We therefore develop the algorithm {\SDIRS} (RangeSearch) for efficient skyline computation with dimension indexing.
Our experimental evaluation shows that {\SDIRS} outperforms our baseline algorithms ({\BNL}, {\SFS}, and {\SALSA}) in general, especially on high-dimensional data.

The remainder of this paper is organized as follows.
Section 2 reviews related skyline computation approaches.
In Section 3, we present our dimension indexing framework and prove several important properties, based on which we propose the algorithm {\SDIRS} in Section 4.
Section 5 reports our experimental evaluation of the performance of {\SDIRS} in comparison with several baseline benchmarks.
Finally, we conclude in Section 6.

\section{Related Work}

In this section, we briefly introduce mainstream skyline computation algorithms.

B\"{o}rzs\"{o}ny et al.\cite{Borzsony2001Operator} first proposed the concept of skyline and several basic computation algorithms, of wihch Nested Loop ({\NL}) is the most straightforward algorithm by comparing each pair of tuples, but always has the the same time complexity $O(n^2)$ no matter the distribution of the data.
Built on top of the naive {\NL} algorithm, Block Nested Loop ({\BNL}) algorithm employees memory window to speed up the efficiency significantly, and of which the best case complexity is reduced to $O(n)$ when there is no temporary file generated during the process of {\BNL}, however the worst case is $O(n^2)$, such as all tuples in database are incomparable with each other.

{\BITMAP} and {\INDEX} \cite{Tan2001IndexBitmap} are two efficient algorithms for skyline computation.
{\BITMAP} based skyline computation is very efficient, however it limits to databases with limited distinct value of each dimension; it also consumes high I/O cost and requires large memory when the database is huge.
{\INDEX} generates the index based on the best value's dimension of tuples.
It is clear that skyline tuples are more likely to be on the top of each index table, so index tables can prune tuple if one tuple's minimum value in all dimensions is larger than the maximal value of all dimensions of another tuple.

Sorted First Skyline ({\SFS}) \cite{Jan2005SFS} and Sort and Limit Skyline algorithm ({\SALSA}) \cite{Bartolini2006SaLSa} are another two pre-sort based algorithms.
{\SFS} has a similar process as {\BNL} but presorts tuples based on the skyline criteria before reading them into window.
{\SALSA} shares the same idea as {\SFS} to presort tuples, but the difference between {\SFS} and {\SALSA} is that they use different approach to optimize the comparison passes: {\SFS} uses entropy function to calculate the probability of one tuple being skyline and {\SALSA} uses stop point.
Indeed, {\SALSA} is designed on top of such an observation: if a skyline tuple can dominate all unread tuples, then the skyline computation can be terminated.
Such a special tuple is called the {\em stop point} in {\SALSA}, which can effectively prune irrelevant tuples that they cannot be in the Skyline.
However, the selection of the stop point depends on dominance comparisons that is completely different from our notion of stop line, which is determined by dimensional indexes without dominance comparison.

{\SUBSKY} algorithm \cite{Tao2007Subsky} converts the $d$-dimensional tuples into 1D value $f(t)$ so all tuples will be sorted based on $f(t)$ value and that helps to determine whether a tuple is dominated by a skyline tuple.
{\SUBSKY} sorts the whole database on full space but calculates skyline on subspace based on user criteria.
Nevertheless, the full space index may not be accurate when pruning data as the index maybe calculated on unrelated dimension.
{\SDI} also supports to calculate skyline on subspace but without re-sorting tuples.
Moreover, dimension index could guarantee the best sorting of subspace and prune more tuples.

Besides sorting based algorithms, there are some algorithms solve the skyline computation problem using R-tree structure, such as {\NN} (Nearest Neighbors) \cite{Kossmann2002NN} and {\BBS} (Branch-and-Bound Skyline) \cite{Papadias2005BBS}.
{\NN} discovers the relationships between nearest neighbors and skyline results.
It is observed that the skyline tuple must be close to the coordinate origin: the tuple which stays closest to the coordinate origin must be a part of the skyline.
Using the first skyline tuple, the database can be further split to several regions, and the first skyline tuple becomes the coordinate origin of these regions.
The nearest point of each region are part of skyline tuples as well, so the whole process iterates until there is no more region split. 
{\BBS} uses the similar idea as {\NN}.
The main difference between {\NN} and {\BBS} is that {\NN} process may include redundant searches but {\BBS} only needs one traversal path.
{\NN} and {\BBS} are both efficient but nevertheless rely on complex data structure which is not necessary for {\SDI} algorithm.

\section{Dimension Indexing for Skyline Computation}

We present in this section the {\SDI} (Skyline on Dimension Index) framework, within which we prove several interesting properties that allow to significantly reduce the total count of dominance comparisons during the skyline computation.

Let $\mathcal D$ be a $d$-dimensional database that contains $n$ tuples, each tuple $t \in \mathcal D$ is a vector of $d$ attributes with $|t| = d$.
We denote $t[i]$, for $1 \leq i \leq d$, the {\em dimensional value} of a tuple $t$ in {\em dimension} $i$ (in the rest of this paper, we consider by default that $i$ satisfies $1 \leq i \leq d$).
Given a total order $\succ_i$ on all values in dimension $i$, we say that the value $t[i]$ of the tuple $t$ is {\em better} than the respective value $t^\prime[i]$ of the tuple $t^\prime$ if and only if $t[i] \succ_i t^\prime[i]$; if $t[i] = t^\prime[i]$, we say that $t[i]$ is {\em equal} to $t^\prime[i]$, and so that $t^\prime[i]$ is {\em not worse} than $t[i]$ if and only if $t[i] \succ_i t^\prime[i] \lor t[i] = t^\prime[i]$, denoted by $t[i] \succeq_i t^\prime[i]$.
Besides, $t^\prime[i]$ is {\em not better} than $t[i]$ is denoted by $t[i] \not\succ t^\prime[i]$
We have that $t[i] \succ_i t^\prime[i] \Rightarrow t[i] \succeq_i t^\prime[i]$.
Without lose of the generality, we denote by the total order $\succ$ the ensemble of all total orders $\succ_i$ on all dimensions and, without confusion, $\{\succ, \succeq, \not\succ\}$ instead of $\{\succ_i, \succeq_i, \not\succ_i\}$.

\begin{definition}[Dominance]
Given the total order $\succ$ and a database $\mathcal D$, a tuple $t \in \mathcal D$ dominates a tuple $t^\prime \in \mathcal D$ if and only if $t[i] \succeq t^\prime[i]$ on each dimension $i$, and $t[k] \succ t^\prime[k]$ for at least one dimension $k$, denoted by $t \succ t^\prime$.
\end{definition}

Further, we denote $t \prec\succ t^\prime \iff t \not\succ t^\prime \land t^\prime \not\succ t$ that the tuple $t$ and the tuple $t^\prime$ are {\em incomparable}.
A tuple is a {\em skyline tuple} if and only if there is no tuple can dominate it.
We therefore formally define {\em skyline} as follows.

\begin{definition}[Skyline]
Given the total order $\succ$ and a database $\mathcal D$, a tuple $t \in \mathcal D$ is a skyline tuple if and only if $\not\exists u \in \mathcal D$ such that $u \succ t$.
The skyline $\mathcal S$ on $\succ$ is the complete set of all skyline tuples that $\mathcal S = \{t \in \mathcal D \mid \not\exists u \in \mathcal D, u \succ t\}$.
\end{definition}

It's easy to see that the skyline $\mathcal S$ of a database $\mathcal D$ is the complete set of all incomparable tuples in $D$, that is, $s \prec\succ t$ for any two tuples $s, t \in \mathcal S$.

\begin{table}[htbp]
\begin{center}
{\scriptsize\begin{tabular}{|r|cccccc|c|}
\hline
ID & $D_1$ & $D_2$ & $D_3$ & $D_4$ & $D_5$ & $D_6$ & Skyline\\
\hline
$t_0$ & 7.5 & 1.3 & 7.5 & 4.5 & 5.3 & 2.1 & {\bf Yes}\\
\hline
$t_1$ & 4.7 & 6.7 & 6.7 & 9.3 & 3.8 & 5.1 & {\bf Yes}\\
\hline
$t_2$ & 8.4 & 9.4 & 5.3 & 5.8 & 6.7 & 7.5 & No\\
\hline
$t_3$ & 5.3 & 6.6 & 6.7 & 6.8 & 5.8 & 9.3 & {\bf Yes}\\
\hline
$t_4$ & 8.4 & 5.2 & 5.1 & 5.5 & 4.1 & 7.5 & {\bf Yes}\\
\hline
$t_5$ & 9.1 & 7.6 & 2.6 & 4.7 & 7.3 & 6.2 & {\bf Yes}\\
\hline
$t_6$ & 5.3 & 7.5 & 1.9 & 5.9 & 3.4 & 1.8 & {\bf Yes}\\
\hline
$t_7$ & 5.3 & 7.5 & 6.7 & 7.2 & 6.3 & 8.8 & No\\
\hline
$t_8$ & 6.7 & 7.3 & 7.6 & 9.7 & 5.3 & 8.7 & No\\
\hline
$t_9$ & 7.5 & 9.6 & 4.8 & 8.9 & 9.5 & 6.5 & No\\
\hline
\end{tabular}}
\end{center}
\caption{A sample database with $d = 6$, $n = 10$, and $m = 6$.}
\label{tab:sample}
\end{table}

Table \ref{tab:sample} shows a sample database of 6 dimensions ($d = 6$) that contains 10 tuples ($n = 10$), of which 6 are skyline tuples ($|\mathcal S| = 6$, we also note the size of Skyline as $m$ with reference to most literature) while the order {\em less than} is applied to all dimensions.

\begin{example}
Among all the 10 tuples $t_0, t_1, \ldots, t_9$ listed in Table \ref{tab:sample}, $t_1 \succ t_8$, $t_4 \succ t_2$, $t_6 \succ t_7$, and $t_6 \succ t_9$; $t_0$, $t_3$, and $t_5$ do not dominate any tuples and are not dominated by any tuples.
The Skyline is therefore $\mathcal S = \{t_0, t_1, t_3, t_4, t_5, t_6\}$.
\qed
\end{example}

The basis of our approach is to build dimensional indexes with respect to the concerned per-dimension total orders that allow to determine the skyline without performing dominance comparisons neither to all tuples in the database nor to all tuples in current skyline.
In general, our approach can significantly reduce the total number of dominance comparisons, which plays an essential role that definitively affects the total processing time of Skyline computation.
Furthermore, our approach constructs the Skyline progressively so no delete operation is required.

For each dimension $i$ of the database $\mathcal D$, the total order $\prec_i$ can be considered as a sorting function $f_i : \mathcal D[i] \rightarrow \mathcal I_i$, where $\mathcal I_i$ is an ordered list of all tuple values in the dimension $i$ of database.
We call such a list $\mathcal I_i$ a {\em dimensional index}.

\begin{definition}[Dimensional Index]
Given a database $\mathcal D$, the dimensional index $\mathcal I_i$ for a dimension $i$ is an ordered list of tuple IDs sorted first by dimensional values with respect to the total order $\succ$, and then, in case of ties, by their lexicographic order.
\end{definition}

In order to avoid unnecessary confusions, we represent a dimensional index $\mathcal I_i$ as a list of entries $\left<t[i]:t.id\right>$ such as which shown in Table \ref{tab:full-di} (where all skyline tuples are in bold).

\begin{table}[htbp]
\begin{center}
{\scriptsize\begin{tabular}{|c|c|c|c|c|c|}
\hline
$\mathcal I_1$ & $\mathcal I_2$ & $\mathcal I_3$ & $\mathcal I_4$ & $\mathcal I_5$ & $\mathcal I_6$\\
\hline
{\bf 4.7:1} & {\bf 1.3:0} & {\bf 1.9:6} & {\bf 4.5:0} & {\bf 3.4:6} & {\bf 1.8:6}\\
{\bf 5.3:3} & {\bf 5.2:4} & {\bf 2.6:5} & {\bf 4.7:5} & {\bf 3.8:1} & {\bf 2.1:0}\\
{\bf 5.3:6} & {\bf 6.6:3} &     {4.8:9} & {\bf 5.5:4} & {\bf 4.1:4} & {\bf 5.1:1}\\
    {5.3:7} & {\bf 6.7:1} & {\bf 5.1:4} &     {5.8:2} & {\bf 5.3:0} & {\bf 6.2:5}\\
    {6.7:8} &     {7.3:8} &     {5.3:2} & {\bf 5.9:6} &     {5.3:8} &     {6.5:9}\\
{\bf 7.5:0} & {\bf 7.5:6} & {\bf 6.7:1} & {\bf 6.8:3} & {\bf 5.8:3} &     {7.5:2}\\
    {7.5:9} &     {7.5:7} & {\bf 6.7:3} &     {7.2:7} &     {6.3:7} & {\bf 7.5:4}\\
    {8.4:2} & {\bf 7.6:5} &     {6.7:7} &     {8.9:9} &     {6.7:2} &     {8.7:8}\\
{\bf 8.4:4} &     {9.4:2} & {\bf 7.5:0} & {\bf 9.3:1} & {\bf 7.3:5} &     {8.8:7}\\
{\bf 9.1:5} &     {9.6:9} &     {7.6:8} &     {9.7:8} &     {9.5:9} & {\bf 9.3:3}\\
\hline
\end{tabular}}
\end{center}
\caption{Dimension indexing of the sample database shown in Table \ref{tab:sample}.}
\label{tab:full-di}
\end{table}

\begin{example}
Table \ref{tab:full-di} shows the 6 dimensional indexes $\mathcal I_1, \mathcal I_2, \ldots, \mathcal I_6$ with respect to all the 6 dimensions $D_I, D_2, \ldots, D_6$ of the sample database shown in Table \ref{tab:sample}.
We show in detail that in $\mathcal I_1$, the dimensional value 5.3 appears in 3 tuples so these 3 entries are secondarily sorted by tuple IDs for $3 < 6 < 7$.
\qed
\end{example}

Now let us consider the dimensional indexes containing distinct dimensional values only, such as $\mathcal I_4$ shown in Table \ref{tab:full-di}.
In such an index $\mathcal I_i$ without duplicate dimensional values, we see that a tuple $t$ can only be dominated by a tuple $s$ such that $o_i(s) < o_i(t)$ (implies that $s[i] < t[i]$ since for any tuple $u$ such that $o_i(u) > o_i(t)$, we have that $t[i] \succ u[i]$ so $u$ cannot dominates $t$.

\begin{lemma}
Given a database $\mathcal D$, let $\mathcal S$ be the skyline of $\mathcal D$, $\mathcal I_i$ be a dimensional index containing only distinct dimensional values, and $t \in \mathcal D$ be a tuple.
Then, $t \in \mathcal S$ if and only if we have $s \not\succ t$ for any skyline tuple $s \in \mathcal S$ such that $o_i(s) < o_i(t)$ on $\mathcal I_i$.
\qed
\label{lem:dist}
\end{lemma}

\begin{proof}
If $o_i(t) = 0$, then $t$ is a skyline tuple because no tuple is better than $t$ in the dimension $D_i$ since all dimensional values are distinct.
If $o_i(t) > 0$, let $s \in \mathcal S$ be a skyline tuple such that $o_i(s) < o_i(t)$, then $s[i] < t[i]$, thus, $s \not\succ t \Rightarrow \exists l \neq i$ such that  $t[l] \succ s[l]$, that is, $s \prec\succ t$; now let $s^\prime \in \mathcal S$ be a skyline tuple such that $o_i(t) < o_i(s^\prime)$, then $s^\prime \in \mathcal S \Rightarrow t \not\succ s^\prime$, further, $t[i] \prec_i s^\prime[i] \Rightarrow s^\prime \not\succ s^\prime$, so we also have $t \prec\succ s^\prime$.
Thus, $t$ is incomparable to any skyline tuple so $t$ is a skyline tuple, that is, $t \in \mathcal S$.
\end{proof}

With Lemma \ref{lem:dist}, to determine whether a tuple $t$ is a skyline tuple, it is only necessary to compare $t$ with each skyline tuple $s$ in one dimension $i$ such that $o_i(s) < o_i(t)$, instead of comparing $t$ with all skyline tuples.
Furthermore, we recall that {\BNL}-like algorithms dynamically update the early skyline set that require a second dominance comparison between an incoming tuple $t$ and early skyline tuple $s$ to determine whether $t \succ s$.
However, with dimensional indexes, Lemma \ref{lem:dist} shows that one dominance comparison $s \not\succ t$ is enough to determine $t \in \mathcal S$, instead of two comparisons.
Lemma \ref{lem:dist} also ensures a {\em progressive} construction of the skyline.

However, in most cases and particularly in real data, there are often duplicate values in each dimension where Lemma \ref{lem:dist} cannot be established.
As shown in Table \ref{tab:full-di}, we can find that there are duplicate values in most dimensions, where a typical instance is $\mathcal I_1$, in which two different cases should be identified:
\begin{enumerate}
\item The dimensional value 5.3 appears in three entries $\left<5.3:3\right>$, $\left<5.3:6\right>$, and $\left<5.3:7\right>$ where $t_3$ and $t_6$ are skyline tuples and $t_7$ is not skyline tuple.
\item The dimensional value 8.4 appears in both of the two entries $\left<8.4:2\right>$ and $\left<8.4:4\right>$, where $t_2$ is not skyline tuple but is indexed before the skyline tuple $t_4$.
\end{enumerate}

In the case (1), a simple straightforward scan on these three dimensional index entries can progressively identify that $t_3$ ($t_1 \not\succ t_3$) and $t_6$ ($t_1 \not\succ t_6$ and $t_3 \not\succ t_6$) are skyline tuples and filter out $t_7$ ($t_6 \succ t_7$).
However, in the case (2), a straightforward scan cannot progressively identify skyline tuples because: there is no precedent tuple dominating $t_2$ so $t_2$ will be first identified as a skyline tuple; then, since no tuple can dominate $t_4$, $t_4$ will identified as a skyline tuple without checking whether $t_4 \not\succ t_2$, hence, finally the output skyline is wrong.

To resolve such misidentifications of skyline tuples, we propose a simple solution that first divides a dimensional index into different logical {\em blocks} of entries with respect to each distinct dimensional value, then apply the {\BNL} algorithm to each {\em block} containing more than one entry to find {\em block skyline tuples} in order to establish Lemma \ref{lem:dist}.

\begin{definition}[Index Block]
Given a database $\mathcal D$, let $\mathcal I_i$ be the dimensional index of a dimension $i$.
An index block of $\mathcal I_i$ is a set of dimensional index entries that share the same dimensional value sorted by the lexicographical order of tuple IDs.
\end{definition}

If each block contains one entry, the only tuple will be compared with existing skyline tuples with respect to Lemma \ref{lem:dist}; otherwise, for any block contains more than one entry, each {\em block skyline tuple} must be compared with existing skyline tuples with respect to Lemma \ref{lem:dist}.
We can generalize the notion of tuples in Lemma \ref{lem:dist} to block skyline tuples because one block contains one entry, the concerned tuples are block skyline tuples.

\begin{theorem}
Given a database $\mathcal D$, let $\mathcal S$ be the Skyline of $\mathcal D$, $\mathcal I_i$ be a dimensional index, and $t \in \mathcal D$ be a block skyline tuple on $\mathcal I_i$.
Then, $t \in \mathcal S$ if and only if we have $s \not\succ t$ for any skyline tuple $s \in \mathcal S$ such that $o_i(s) < o_i(t)$ on $\mathcal I_i$.
\label{the:sdi}
\end{theorem}

\begin{proof}
With the proof of Lemma \ref{lem:dist} and the statement of block skyline tuples, the proof of Theorem \ref{the:sdi} is immediate.
\end{proof}

\begin{table}[htbp]
\begin{center}
{\scriptsize\begin{tabular}{|c|c|c|c|}
\hline
\multicolumn{4}{|l|}{$\mathcal I_1$}\\
\hline
\hline
4.7:1 & \multicolumn{3}{l|}{}\\
\hline
{\bf 5.3:3} & {\bf 5.3:6} & 5.3:7 & ~~~~~~~~\\
\hline
6.7:8 & \multicolumn{3}{l|}{}\\
\hline
{\bf 7.5:0} & {\bf 7.5:9} & \multicolumn{2}{l|}{}\\
\hline
8.4:2 & {\bf 8.4:4} & \multicolumn{2}{l|}{}\\
\hline
9.1:5 & \multicolumn{3}{l|}{}\\
\hline
\end{tabular}}
\end{center}
\caption{A block view of the dimensional index $\mathcal I_i$.}
\label{tab:block}
\end{table}

\begin{example}
As shown in Table \ref{tab:block}, 6 blocks can be located from $\mathcal I_1$ with respect to all 6 distinct values: 4.7, 5.3, 6.7, 7.5, 8.4, and 9.1.
According to Theorem \ref{the:sdi}:
the block 4.7 contains $t_1$, so $t_1$ is a block skyline tuple and is the first skyline tuple;
the block 5.3 contains $t_3$, $t_6$, and $t_7$ where $t_3 \prec\succ t_6$ and $t_6 \succ t_7$, so $t_3$ and $t_6$ block skyline tuples such that $t_1 \prec\succ t_3$ and $t_1 \prec\succ t_6$, hence, $t_3$ and $t_6$ are new skyline tuples;
the block 6.7 contains $t_8$, so $t_8$ is a block skyline tuple that is dominated by $t_6$;
the block 7.5 is different from the block 5.3, where $t_0 \prec\succ t_9$ so both of them are block skylines, and we have $t_6 \succ t_9$ so $t_0$ is a skyline tuple;
the block 8.3 is the same case as the block 5.3, where $t_4$ is a skyline tuple;
finally, no skyline tuple dominates $t_5$, so the Skyline is $\{t_0, t_1, t_3, t_4, t_5, t_6\}$.
\qed
\end{example}

It is important to note that Theorem \ref{the:sdi} allows dominance comparisons to be performed on arbitrary dimensional indexes and the computation stops while the last entry in any index is reached.
Therefore, we see that a dynamic {\em dimension switching} strategy can further improve the efficiency of the Skyline computing based on dimension indexing.
For instance, if we proceed a breadth-first search strategy among all dimensional indexes shown in Table \ref{tab:full-di}, while we examine the second entry $\left<2.6:5\right>$ in $\mathcal I_3$, although currently $\mathcal S = \{t_0, t_1, t_6, t_3, t_4\}$, we do not have to compare $t_5$ with all those skyline tuples but only with $t_6$; if we continue to examine the second entry in $\mathcal I_4$, $t_5$ can be ignored since it is already a skyline tuple.
We also note that duplicate dimensional values present in tuples severely impact the overall performance of dimensional index based Skyline computation, therefore, reasonable dimension selection/sorting heuristics shall be helpful.

\section{A Range Search Approach to Skyline}

In this section, we first propose the notion of {\em stop line} that allows terminate searching skyline tuples by pruning non relevant tuples, then present the algorithm {\SDIRS} ({\RS}) for skyline computation based on the {\SDI} framework.
Notice that the name {\RS} stands for the bounded search range while determining skyline tuples.

\subsection{Stop Line}

Let us consider again the Skyline and the dimensional indexes shown in Table \ref{tab:full-di}.
It is easy to see that all 6 skyline tuples can be found at the first two entries of all dimensional indexes, hence, a realistic question is whether we can stop the Skyline computation before reaching the end of any dimensional index.

\begin{definition}[Stop Line]
Given a database $\mathcal D$, let $p \in \mathcal D$ be a skyline tuple. A stop line established from $p$, denoted by $S_p$, is a set of dimensional index entries $\left<p[i]:p\right>$ such that $p$ appears in each dimension.
An index entry $e \S_p$ is a stop line entry and an index block containing a stop line entry is a stop line block.
\end{definition}

Let $t$ be a tuple, we denote $b_i(t)$ the offset of the index block on a dimensional index $\mathcal I_i$ that contains $t$, that is, the position of the index block that contains $t$.
Hence, let $p$ be a stop line tuple and $t$ be a tuple, we say that the stop line $S_p$ {\em covers} the index entry $\left<t[i]:t\right>$ on a dimensional index $\mathcal I_i$ if $b_i(p) < b_i(t)$.
For instance, Table \ref{tab:sl-6} shows the stop line created from the tuple $t_6$, which totally covers 41 index entries without $\left<5.3:7\right>$ on $\mathcal I_1$ neither $\left<7.5:7\right>$ on $\mathcal I_2$.

\begin{table}[htbp]
\begin{center}
{\scriptsize\begin{tabular}{|c|c|c|c|c|c|}
\hline
$\mathcal I_1$ & $\mathcal I_2$ & $\mathcal I_3$ & $\mathcal I_4$ & $\mathcal I_5$ & $\mathcal I_6$\\
\hline
{\bf 4.7:1} & {\bf 1.3:0} & \underline{\bf 1.9:6} & {\bf 4.5:0} & \underline{\bf 3.4:6} & \underline{\bf 1.8:6}\\
{\bf 5.3:3} & {\bf 5.2:4} & {\bf 2.6:5} & {\bf 4.7:5} & {\bf 3.8:1} & {\bf 2.1:0}\\
\underline{\bf 5.3:6} & {\bf 6.6:3} &     {4.8:9} & {\bf 5.5:4} & {\bf 4.1:4} & {\bf 5.1:1}\\
    {5.3:7} & {\bf 6.7:1} & {\bf 5.1:4} &     {5.8:2} & {\bf 5.3:0} & {\bf 6.2:5}\\
    {6.7:8} &     {7.3:8} &     {5.3:2} & \underline{\bf 5.9:6} &     {5.3:8} &     {6.5:9}\\
{\bf 7.5:0} & \underline{\bf 7.5:6} & {\bf 6.7:1} & {\bf 6.8:3} & {\bf 5.8:3} &     {7.5:2}\\
    {7.5:9} &     {7.5:7} & {\bf 6.7:3} &     {7.2:7} &     {6.3:7} & {\bf 7.5:4}\\
    {8.4:2} & {\bf 7.6:5} &     {6.7:7} &     {8.9:9} &     {6.7:2} &     {8.7:8}\\
{\bf 8.4:4} &     {9.4:2} & {\bf 7.5:0} & {\bf 9.3:1} & {\bf 7.3:5} &     {8.8:7}\\
{\bf 9.1:5} &     {9.6:9} &     {7.6:8} &     {9.7:8} &     {9.5:9} & {\bf 9.3:3}\\
\hline
\end{tabular}}
\end{center}
\caption{The stop line created from tuple $t_6$ covers 41 index entries in total.}
\label{tab:sl-6}
\end{table}

Obviously, let $p$ be a stop line tuple and $t$ be a tuple such that $p \prec t$, then we have that $b_i(p) \leq b_i(t)$ on each dimensional index $\mathcal I_i$ and $b_k(p) < b_k(t)$ on at least one dimensional index $\mathcal I_k$.

\begin{theorem}
Given a database $\mathcal D$, let $S_p$ be the stop line with respect to a skyline tuple $p$.
By following any top-down traversal of all dimensional indexes, if all stop line blocks have been traversed, then the complete set of all skyline tuples has been generated and the skyline computation can stop.
\label{the:stp}
\end{theorem}

\begin{proof}
Let $p$ be a skyline tuple and $t \in \mathcal S \setminus p$ be a skyline tuple, we have (1) $t \prec\succ p$ or (2) $t = p$, if $t$ and $p$ have identical dimensional values.
In the first case, $t \prec\succ p \Rightarrow \exists k, p[k] \succ t[k] \Rightarrow b_k(p) < b_k(t)$, that is, if the index traversal passes the stop line $S_p$, the tuple $t$ must have been identified at least in the dimension $D_k$.
In the second case, we have $b_i(p) = b_i(t)$ for any dimension $D_i$.
In both cases, if all stop line blocks have been processed, then all skyline tuples have been found.
\end{proof}

In principle, any skyline tuple can be chosen to form a stop line, however, different stop lines behave differently in pruning useless tuples.
For instance, as shown in Table \ref{tab:sl-6}, the stop line $S_6$ created from $t_6$ covers totally 41 index entries and two tuples $\{t_7, t_9\}$ can be pruned; however, as shown in Table \ref{tab:sl-0}, the the stop line $S_0$ created from $t_0$ covers only 37 index entries and no tuple can be pruned.
Obviously, a good stop line shall cover index entries at much as possible, so we can use an optimal function, $min_p$, to minimize the offsets of a skyline tuple in all dimensional indexes for building a stop line $S_{p}$, defined as:
\begin{align*}
S_{p} = \mathop{\arg\min}_{p}(max\{o_i(p)\}, \sum_{i = 1}^{d} o_i(p) / d)
\end{align*}
The function min(p) sorts tuples first by the maximum offset, then by the mean offset in all dimensional indexes, so the minimized skyline tuple is the best stop line tuple.
Hence, a dynamically updated stop tuple $p$ can be maintained by keeping $min(p) < min(t)$ for any new skyline tuple $t$.

\begin{table}[htbp]
\begin{center}
{\scriptsize\begin{tabular}{|c|c|c|c|c|c|}
\hline
$\mathcal I_1$ & $\mathcal I_2$ & $\mathcal I_3$ & $\mathcal I_4$ & $\mathcal I_5$ & $\mathcal I_6$\\
\hline
{\bf 4.7:1} & \underline{\bf 1.3:0} & {\bf 1.9:6} & \underline{\bf 4.5:0} & {\bf 3.4:6} & {\bf 1.8:6}\\
{\bf 5.3:3} & {\bf 5.2:4} & {\bf 2.6:5} & {\bf 4.7:5} & {\bf 3.8:1} & \underline{\bf 2.1:0}\\
{\bf 5.3:6} & {\bf 6.6:3} &     {4.8:9} & {\bf 5.5:4} & {\bf 4.1:4} & {\bf 5.1:1}\\
    {5.3:7} & {\bf 6.7:1} & {\bf 5.1:4} &     {5.8:2} & \underline{\bf 5.3:0} & {\bf 6.2:5}\\
    {6.7:8} &     {7.3:8} &     {5.3:2} & {\bf 5.9:6} &     {5.3:8} &     {6.5:9}\\
\underline{\bf 7.5:0} & {\bf 7.5:6} & {\bf 6.7:1} & {\bf 6.8:3} & {\bf 5.8:3} &     {7.5:2}\\
    {7.5:9} &     {7.5:7} & {\bf 6.7:3} &     {7.2:7} &     {6.3:7} & {\bf 7.5:4}\\
    {8.4:2} & {\bf 7.6:5} &     {6.7:7} &     {8.9:9} &     {6.7:2} &     {8.7:8}\\
{\bf 8.4:4} &     {9.4:2} & \underline{\bf 7.5:0} & {\bf 9.3:1} & {\bf 7.3:5} &     {8.8:7}\\
{\bf 9.1:5} &     {9.6:9} &     {7.6:8} &     {9.7:8} &     {9.5:9} & {\bf 9.3:3}\\
\hline
\end{tabular}}
\end{center}
\caption{The stop line created from tuple $t_0$ covers 37 index entries in total.}
\label{tab:sl-0}
\end{table}

Nevertheless, the use of stop lines requires that all stop line blocks in all dimensions being examined, so it is difficult to judge whether a scan reaches first at the end of any dimensional index or first finishes to examine all stop line blocks although we can state that the setting of stop lines can effectively help the Skyline computation in correlated data.
We also note that the use of stop lines require that all dimensions are indexed, which is an additional constraint while applying Theorem \ref{the:sdi} and Theorem \ref{the:stp} together since Theorem \ref{the:sdi} does not impose that all dimensions must be constructed.
We propose, thus, to consider different application strategies of Theorem \ref{the:sdi} and Theorem \ref{the:stp} with respect to particular use cases and data types to accelerate the Skyline computation.

\subsection{The RangeSearch Algorithm}

Theorem \ref{the:sdi} allows to reduce the count of dominance comparisons while computing the skyline.
However, as mentioned in Section 3, the duplicate dimensional values severely augment the dominance comparisons count because a {\BNL} based local comparisons must be applied.
Notice that it is useless to apply {\SFS} or {\SALSA} to such local comparisons because their settings of sorting functions disable one of the most important features of our dimension indexing based approach: individual criterion including that for non-numerical values of skyline selection can be independently applied to each dimension.

In order to reduce the impact of duplicate dimensional values, we propose a simple solution based on sorting dimensional indexes by their cardinalities $|\mathcal I_i|$.
The computation starts from the best dimensional index so the calls of {\BNL} can be minimized.
For instance, in Table \ref{tab:full-di}, all dimensional indexes can be sorted as $|\mathcal I_4| > |\mathcal I_2| = |\mathcal I_5| = |\mathcal I_6| > |\mathcal I_3| > |\mathcal I_1|$, where the best dimensional index $\mathcal I_4$ contains no duplicate dimensional values so Lemma \ref{lem:dist} can be directly established so dimension switching can be performed earlier.

We present then {\SDIRS} (RangeSearch), an algorithm with the application of Theorem \ref{the:sdi} and Theorem \ref{the:stp} by performing dominance comparisons only with a range of skyline tuples instead of all, as shown in Algorithm \ref{algo:rs}, to the skyline computation on sorted dimensional indexes.

\begin{algorithm}[htbp]
\SetKw{And}{and}
\SetKw{Break}{break}
\SetKw{Or}{or}
\KwIn{Sorted dimensional indexes $\mathcal I_\mathcal D$}
\KwOut{Complete set $\mathcal S$ of all skyline tuples}
$L \leftarrow$ empty stop line\\
\While{true} {
    \ForEach{$\mathcal I_i \in \mathcal I_\mathcal D$} {
        \While{$B \leftarrow$ get next block from $\mathcal I_i$}{
            \If{$B = null$}{
                \Return{$\mathcal S$}\\
            }
            \ForEach{$t \in B$}{
                \If{$t$ has been compared \And $t \not\in \mathcal S$}{
                    remove $t$ from $B$\\
                }
            }
            $\mathcal S_B \leftarrow$ compute the block Skyline from $B$ by {\BNL}\\
            \ForEach{$t \in \mathcal S_B$ \And $t \not\in \mathcal S$}{
                \If{$\mathcal S_i \not\prec t$}{
                    $\mathcal S_i \leftarrow \mathcal S_i \cup t$\\
                    $\mathcal S \leftarrow \mathcal S \cup t$\\
                    $L_t \leftarrow$ build stop line from $t$\\
                    \If{$L = \emptyset$ \Or $L_t$ is better than $L$}{
                        $L \leftarrow L_t$\\
                    }
                }
            }
            \If{$o_d \geq L[d]$ for each dimension $d$}{
                \Return{$\mathcal S$}\\
            }
           \If{\mbox{\tt [dimension-switching]}}{
                \Break\\
           }
        }
    }
}
\caption{{\SDIRS} ({\RS})}
\label{algo:rs}
\end{algorithm}

The algorithm accepts a set of sorted dimensional indexes $\mathcal I_\mathcal D$ of a $d$-dimensional database $\mathcal D$ as input and outputs the complete set $\mathcal S$ of all skyline tuples.
First, we initialize an empty stop line $L$, then we enter a Round Robin loop that find the complete set of all skyline tuples with respect to Theorem \ref{the:sdi} and Theorem \ref{the:stp}.
In each dimensional index $\mathcal I_i$ based iteration, we first get the next block $B$ of index entries from $\mathcal I_i$.
According to Theorem \ref{the:sdi}, if $B$ is null, which means that the end of $\mathcal I_i$ is reached, we exit the algorithm by returning $\mathcal S$; otherwise, we treat all index entries block by block to find skyline tuples.
If a tuple $t \in B$ is already compared and marked as non skyline tuple, we should ignore it in order to prevent comparing it with other tuples again; however, if $t$ is a skyline tuple, we shell keep it because $t$ may dominate other new tuples in block-based {\BNL} while computing the block Skyline $\mathcal S_B$.
Therefore, for each tuple $t \in \mathcal S_B$ such that $t \not\in \mathcal S$ (again, we do not want to compare a skyline tuple with other skyline tuples), we compare it with all existing skyline tuples $\mathcal S_i$ present in current dimension $D_i$.
Here we introduce a shortcut operator $\mathcal S_i \not\prec t$ at line 12 that means that none of skyline tuples in $\mathcal S_i$ dominates $t$, and according to Theorem \ref{the:sdi}, $t$ must be a skyline tuple in this case and must be added to the dimensional Skyline $\mathcal S_i$ and the global Skyline $\mathcal S$.
Furthermore, with respect to Theorem \ref{the:stp}, we build a new stop line $L_t$ from each new skyline tuple $t$ and if it is better than current stop line $L$ (or no stop line is defined), we update $L$ by $L_t$.
While the above dominance comparisons are finished, we compare current dimensional iteration position on all dimensions with the latest stop line, if in each dimension the stop line entry is reached, {\RS} stops by returning the complete Skyline $\mathcal S$.
Otherwise, {\RS} switch to the next dimension and repeat the above procedure with respect to a particular {\tt [dimension-switching]} strategy.

In our approach, we consider {\em breadth-first dimension switching} ({\BFS}) and {\em depth-first dimension switching} ({\DFS}).
With {\BFS}, if a block is examined and if {\SDIRS} shall continue to run, then the next dimension will be token.
However, in depth-first switching, if a block is examined and if {\SDIRS} shall continue to run, {\SDIRS} continues to go ahead in current dimension if current block contains new skyline tuples, till to meet a block without any new skyline tuple.
The difference between breadth-first switching and depth-first switching is clear.
{\DFS} tries to accelerate skyline tuple searching in each dimension, this strategy benefits the most from Theorem \ref{the:sdi}; furthermore, if the best stop line is balanced in each dimension, then {\DFS} reaches well the stop line in each dimension so more tuples can be pruned.
However, {\DFS} is not efficient if there are a large number of duplicate values in some dimensions because each block shall be examined before switching to the next dimension.
In this case, depth-first switching takes duplicate dimensional values into account: since all dimensional indexes are sorted with respect to their cardinalities, {\SDIRS} starts always from the best dimensions that contain less duplicate values and finds skyline tuples as much as possible by depth-first switching, hence, while switching to other dimensions, it is possible that some tuples in some blocks have already been compared or are already skyline tuples so no more comparisons will be performed.

In comparison with sorting based algorithms like {\SFS} and {\SALSA}, {\SDIRS} allows to sort tuples with respect to each dimension, which is interesting while different criteria are applied to determine the skyline.
For instance, we can specify the order {\em less than} ($<$) a one dimension and the order {\em greater than} ($>$) to another dimension, without of additional calculation to unifying and normalizing dimensional values.
With the same reason, {\SDIRS} allows to directly process categorical data as numerical data: if any total order can be defined to a categorical attribute, for instance, the user preference on colors such that {\tt blue} $\succ$ {\tt green} $\succ$ {\tt yellow} $\succ$ {\tt red}, then {\SDIRS} can treat such values as any ordered numerical values without any adaptation.

With dimensional indexes, {\SDIRS} is efficient in both space and time complexities.
The storage requirement for dimensional indexes is guaranteed: for instance, a C/C++ implementation of {\SDIRS} may consider an index entry as a {\tt struct} of tuple ID and dimensional value that requires 16 bytes (64bit ID and 64bit value), therefore if each dimensional index corresponds to a {\tt std::vector} structure, the in-memory storage size of dimensional indexes is the double of the database size: for instance, 16GB heap memory fits the allocation of 1,000,000,000 structures of ID/value, as 12,500,000 8-dimensional tuples.
Let $d$ be the dimensionality, $n$ be the cardinality of data, and $m$ be the size of the skyline.
The generation of dimensional indexes requires $\mathcal O(dn\lg{n})$ with respect to a general-purpose sorting algorithm of $\mathcal O(n\lg{n})$ complexity.
For the best-case, that is, $m = 1$, {\SDIRS} finishes in $\mathcal O(1)$ since the only skyline tuple is the stop line and the computation stops immediately; for the worst-case, all $n$ tuples are skyline tuples, {\SDIRS} finishes in $$\mathcal O(\dfrac{n(n - 1)}{2})$$ according to Theorem \ref{the:sdi} if each block contains only one tuple (that is, the case without duplicate dimensional values).
More generally, if the best dimension index contains $k$ duplicate values, then {\SDIRS} finishes in $$\mathcal O(k^2 + \dfrac{(n - k)(n -k - 1)}{2})$$ since the worst-case is that all $k$ duplicate values appear in the same block.

\section{Experimental Evaluation}

In this section, we report our experimental results on performance evaluation of {\SDIRS} that is conducted with both of {\BFS} and {\DFS} dimension switching, and is compared with three baseline algorithms {\BNL}, {\SFS}, and {\SALSA} on synthetic and real benchmark datasets.
The {\tt vol} sorting function and the {\tt max} sorting function are respectively applied to {\SFS} and {\SALSA} as mentioned in \cite{Bartolini2006SaLSa}.

\begin{figure*}[t]
\begin{center}
{\scriptsize\begin{tabular}{cccc}
\multicolumn{2}{c}{Run-time on {\em independent} datasets.} & \multicolumn{2}{c}{Dominance comparisons on {\em independent} datasets.}\\
\includegraphics[width=4cm]{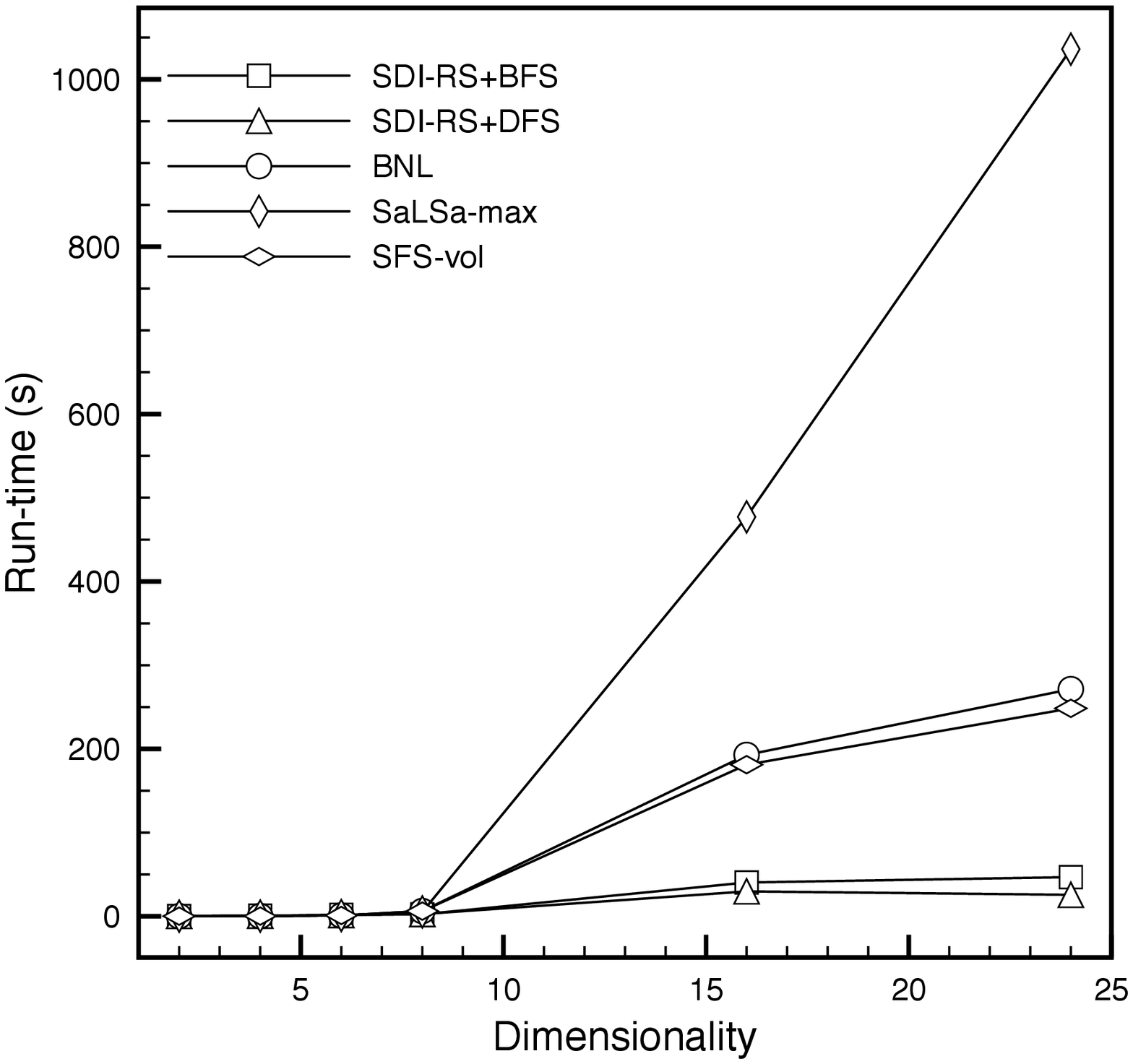}
&
\includegraphics[width=4cm]{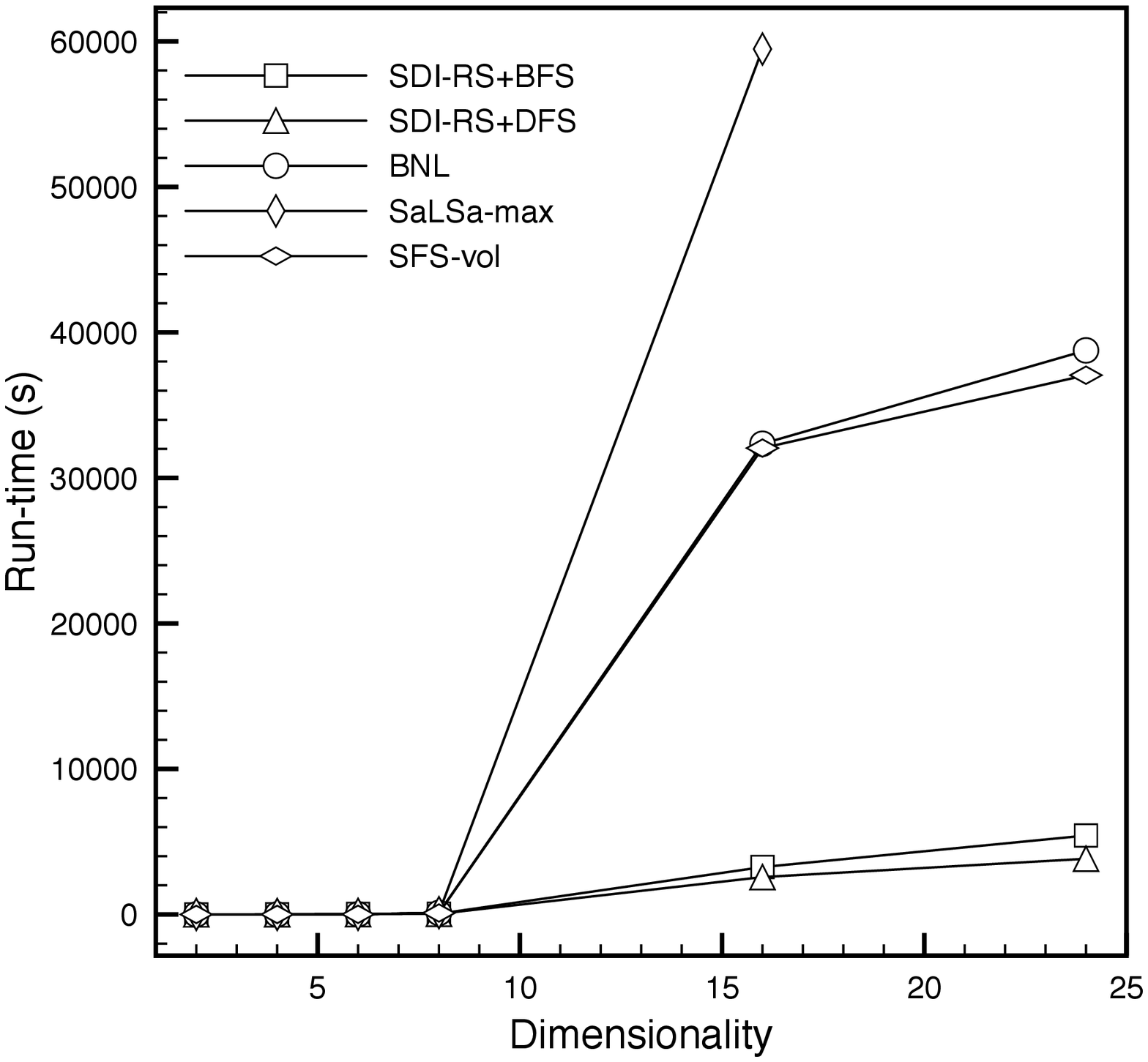}
&
\includegraphics[width=4cm]{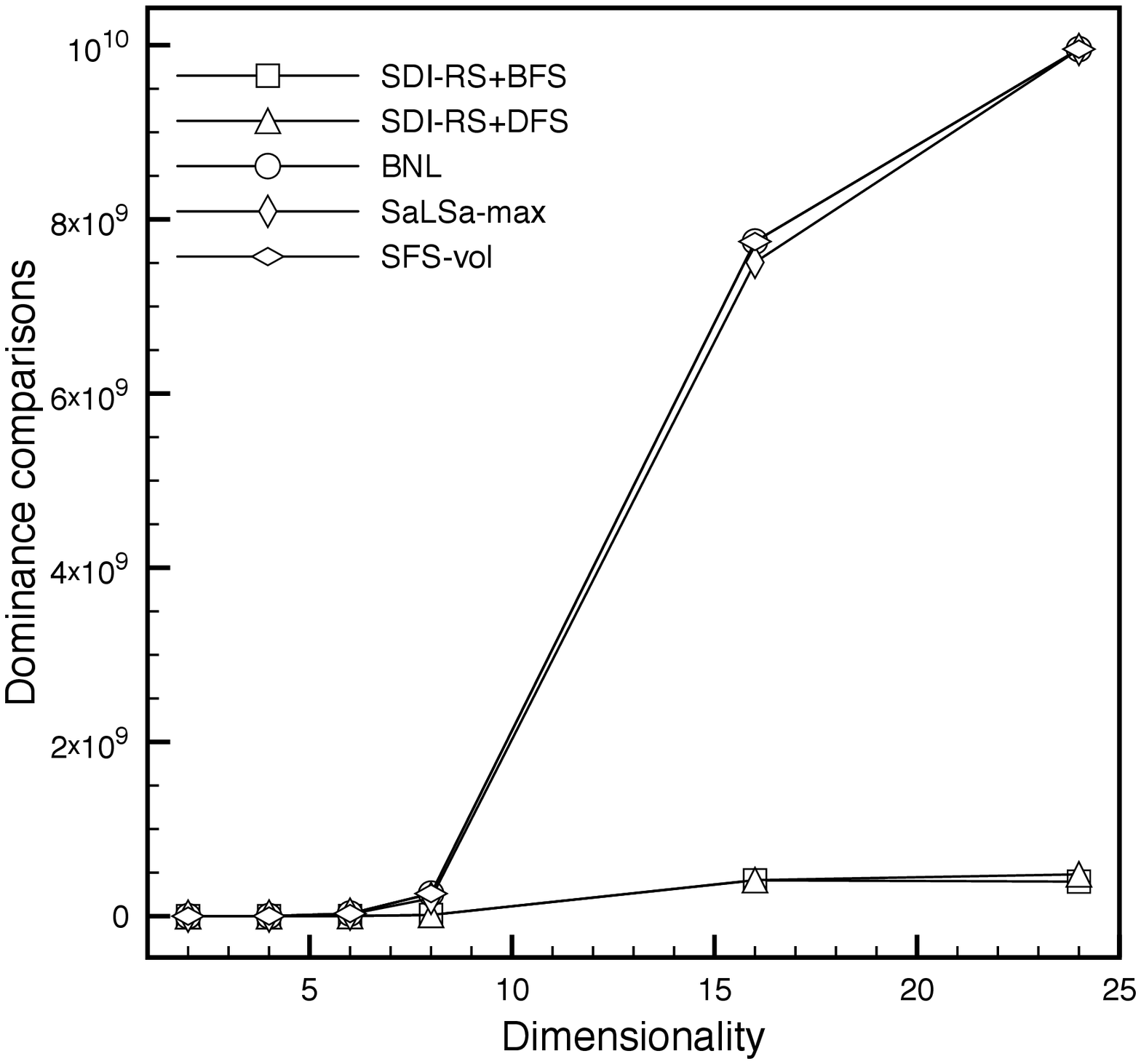}
&
\includegraphics[width=4cm]{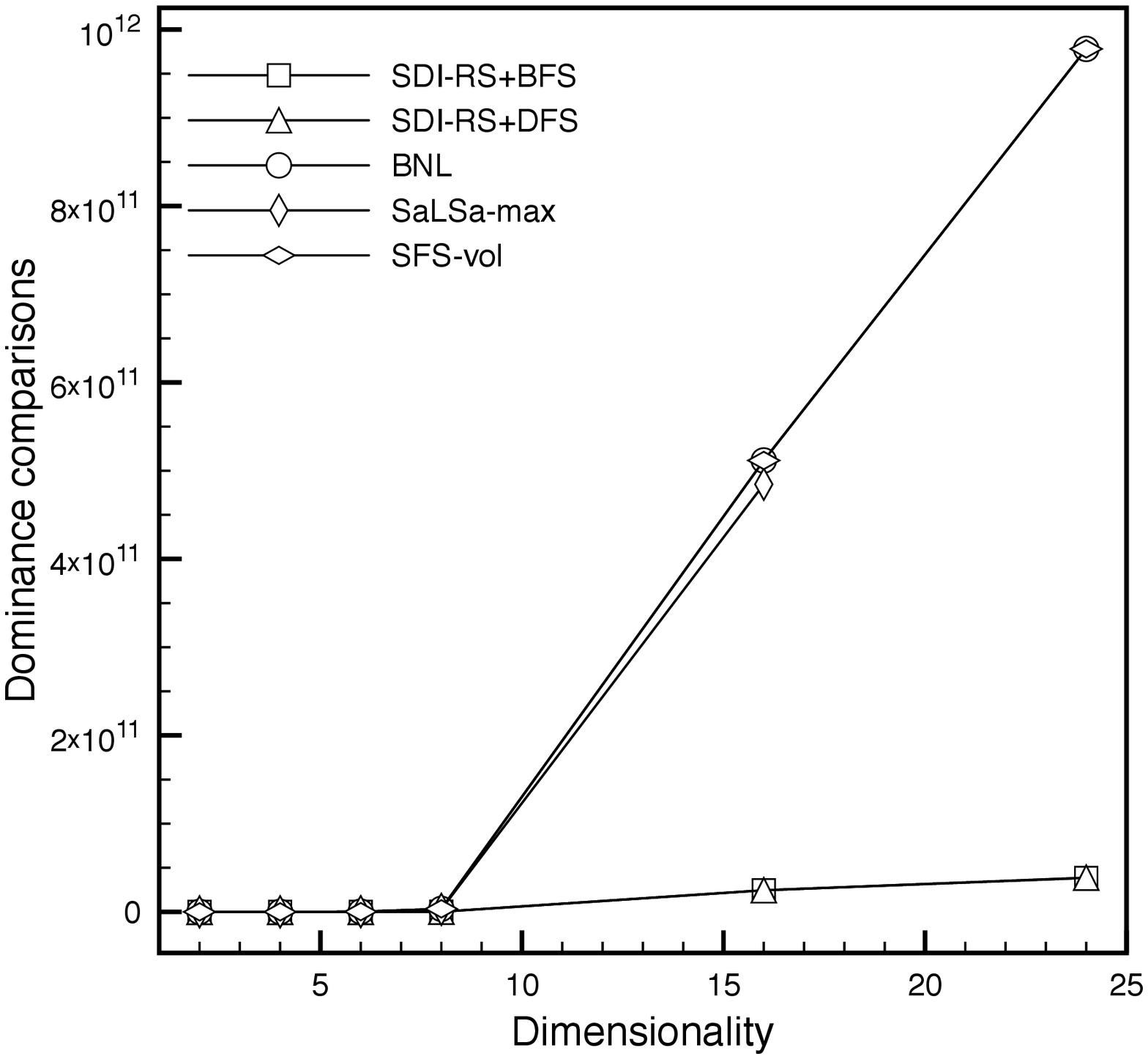}\\
(a) $n = 100$K & (b) $n = 1$M & (c) $n = 100$K & (d) $n = 1$M\\
&&&\\
\multicolumn{2}{c}{Run-time on {\em correlated} datasets.} & \multicolumn{2}{c}{Dominance comparisons on {\em correlated} datasets.}\\
\includegraphics[width=4cm]{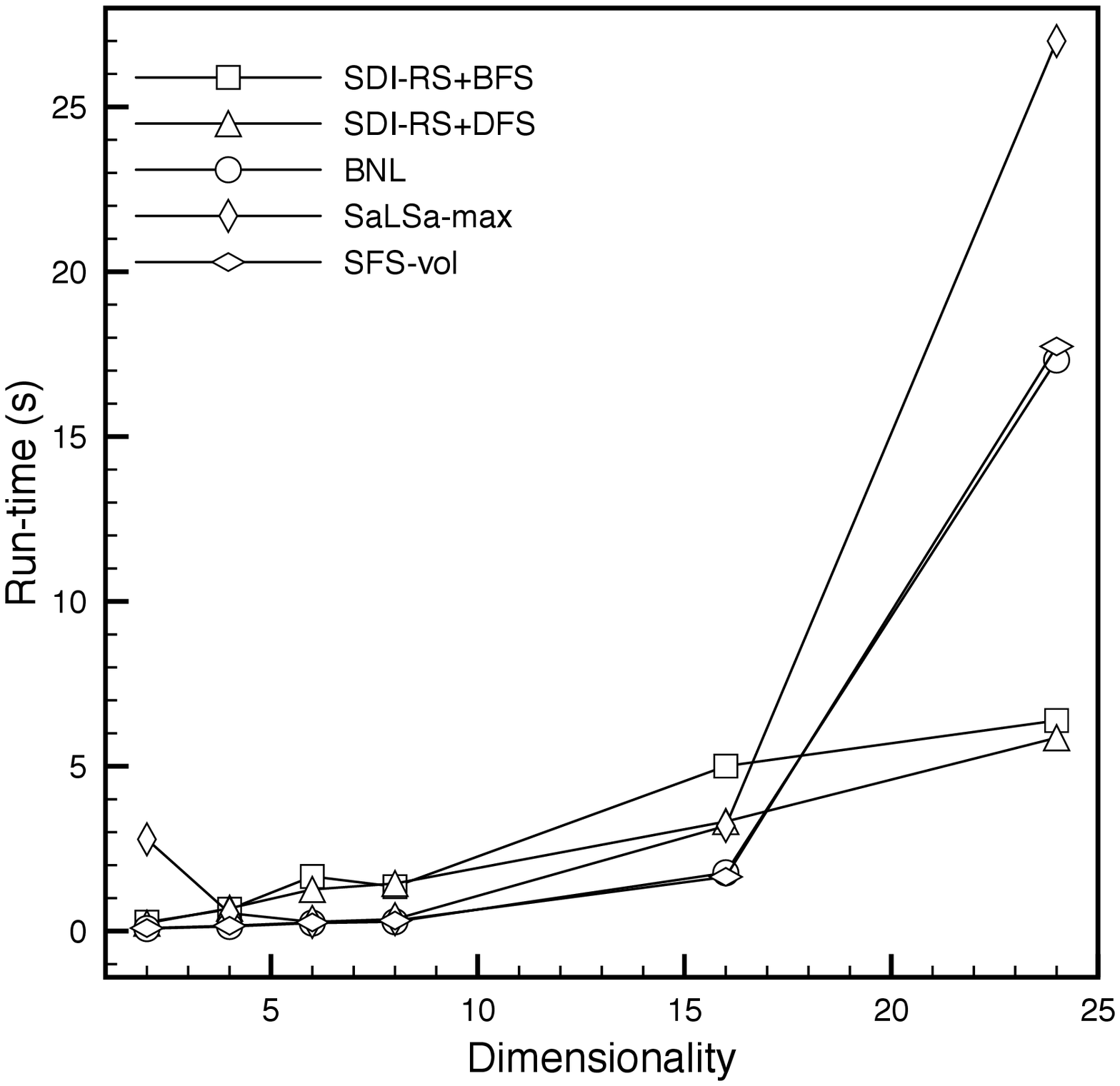}
&
\includegraphics[width=4cm]{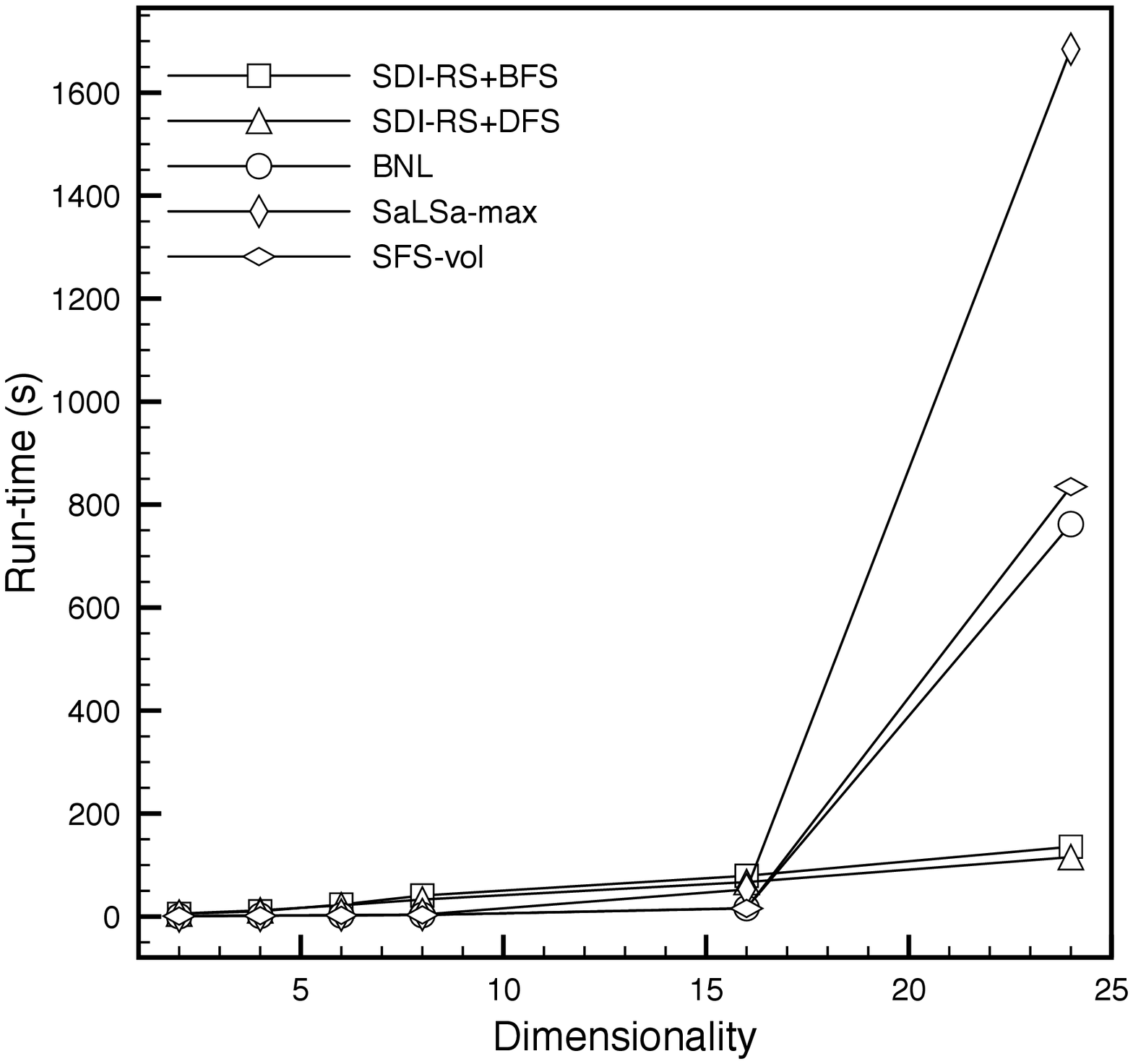}
&
\includegraphics[width=4cm]{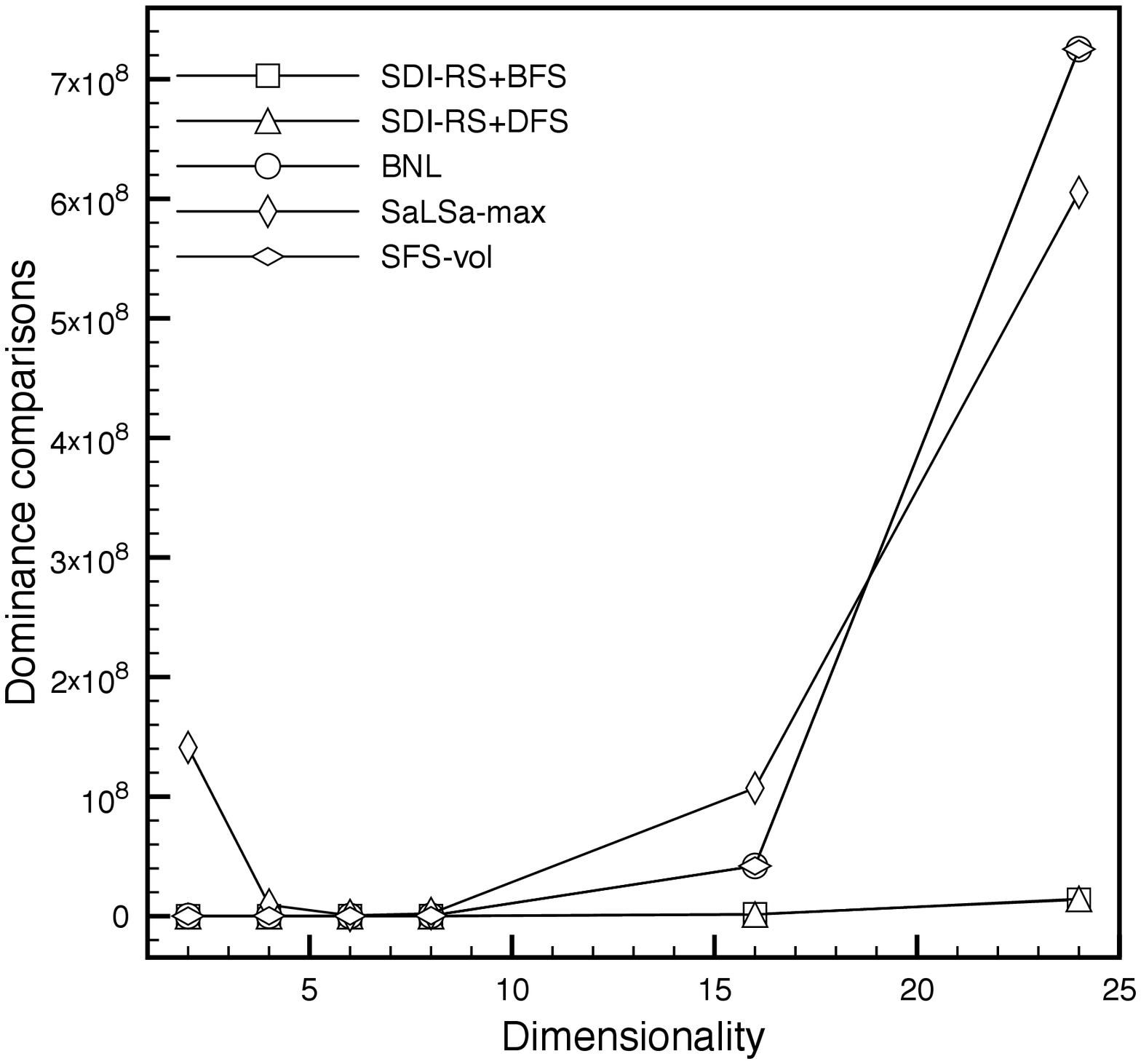}
&
\includegraphics[width=4cm]{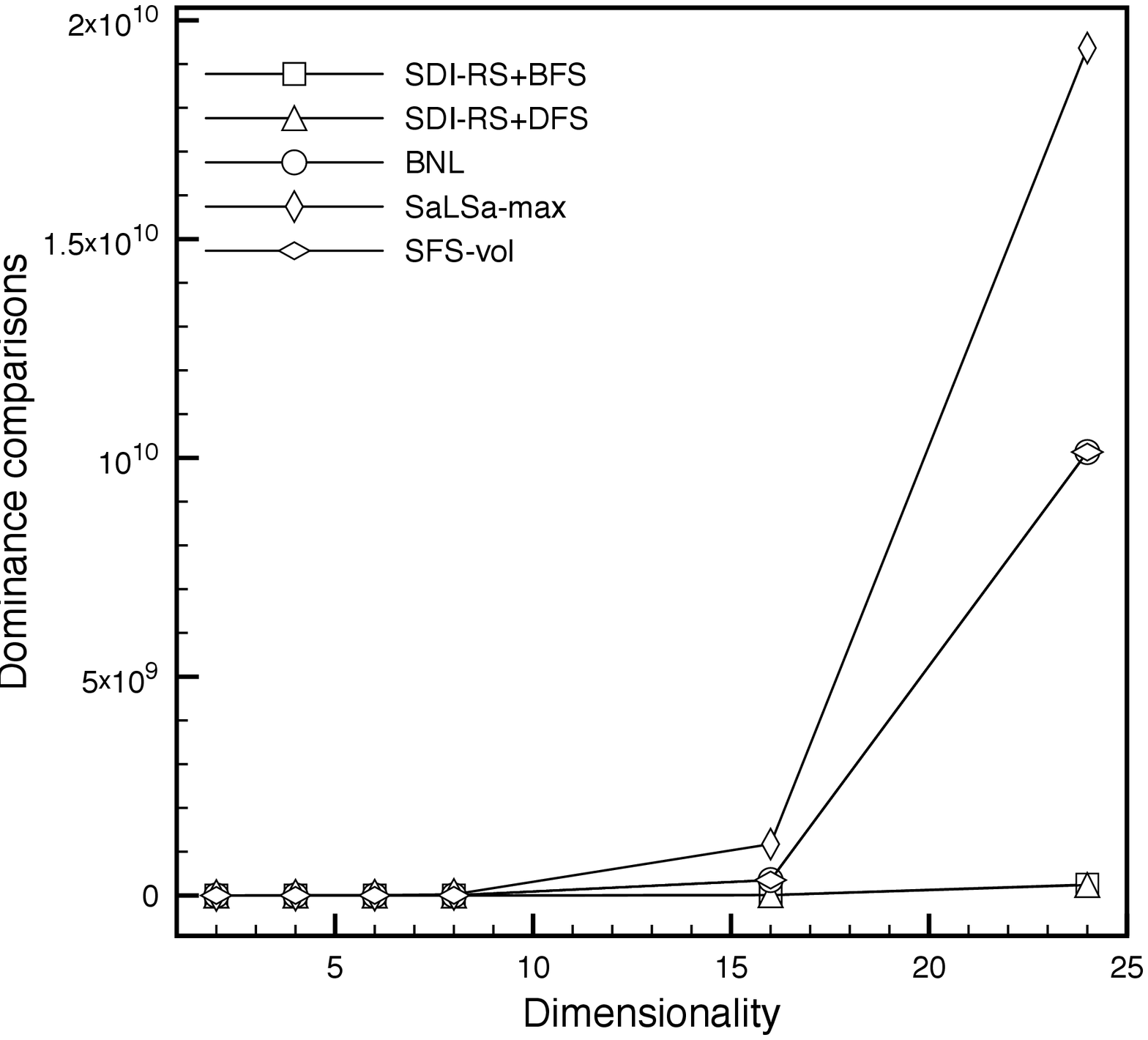}\\
(e) $n = 100$K & (f) $n = 1$M & (g) $n = 100$K & (h) $n = 1$M\\
&&&\\
\multicolumn{2}{c}{Run-time on {\em anti-correlated} datasets.} & \multicolumn{2}{c}{Dominance comparisons on {\em anti-correlated} datasets.}\\
\includegraphics[width=4cm]{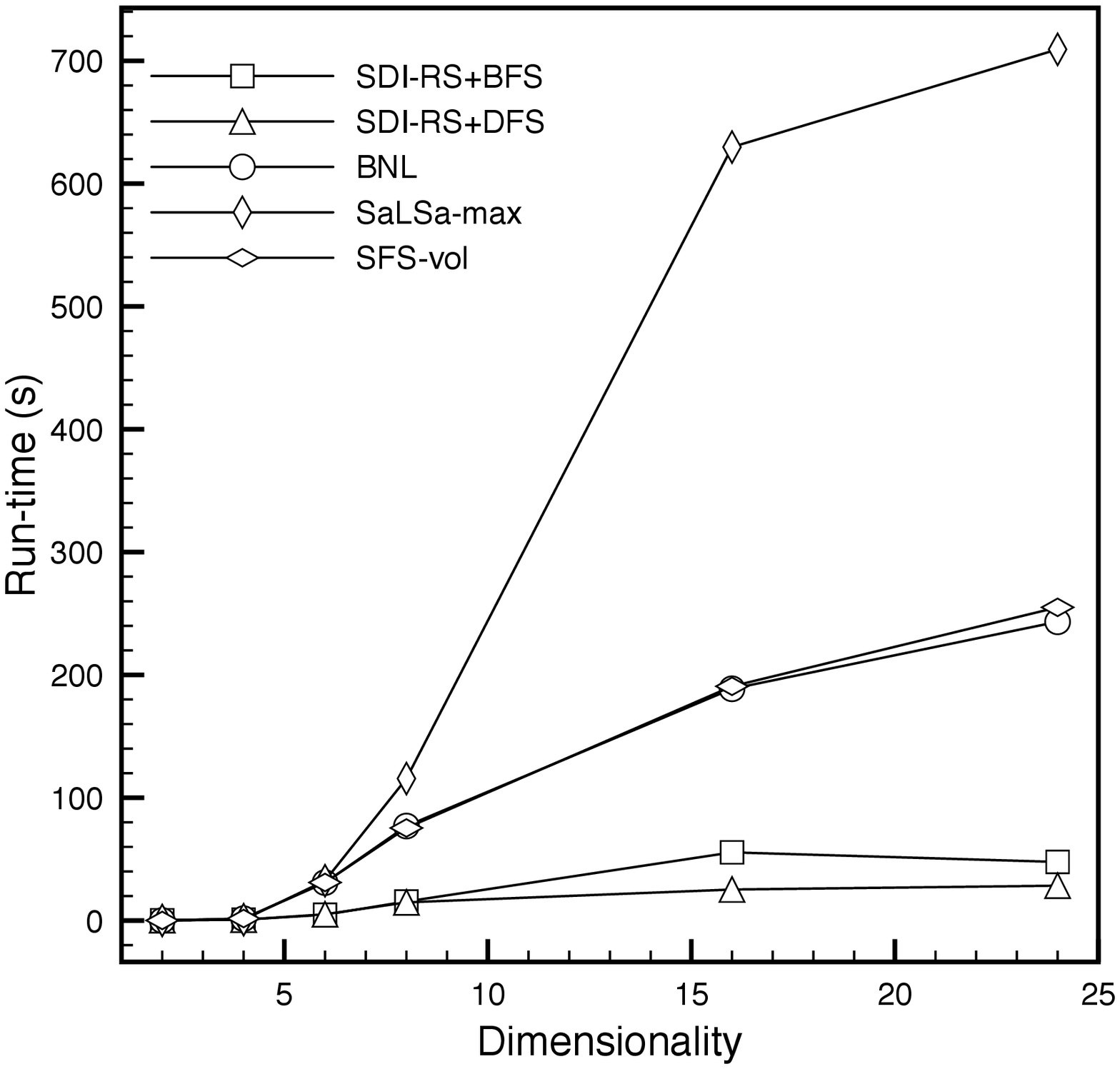}
&
\includegraphics[width=4cm]{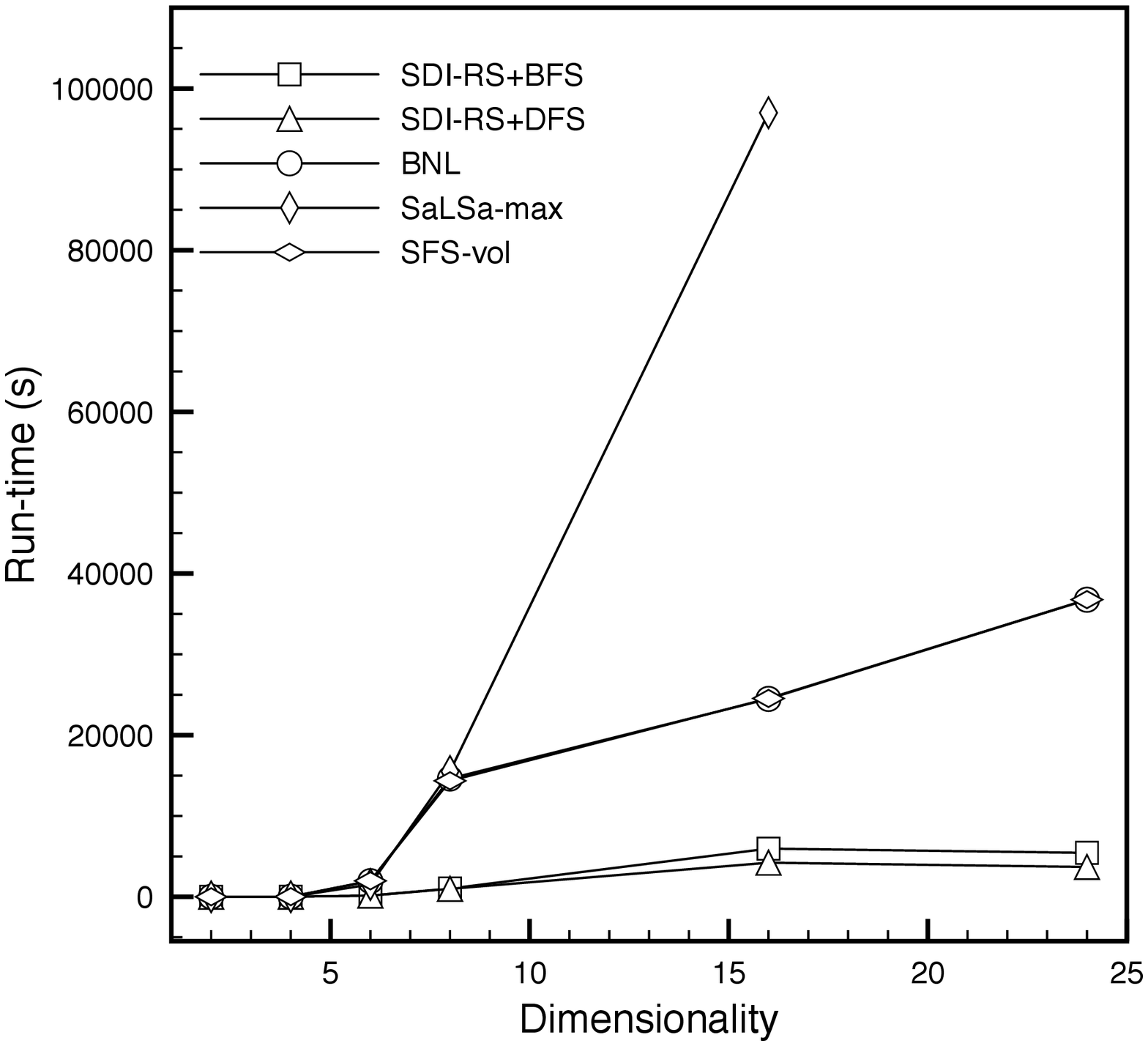}
&
\includegraphics[width=4cm]{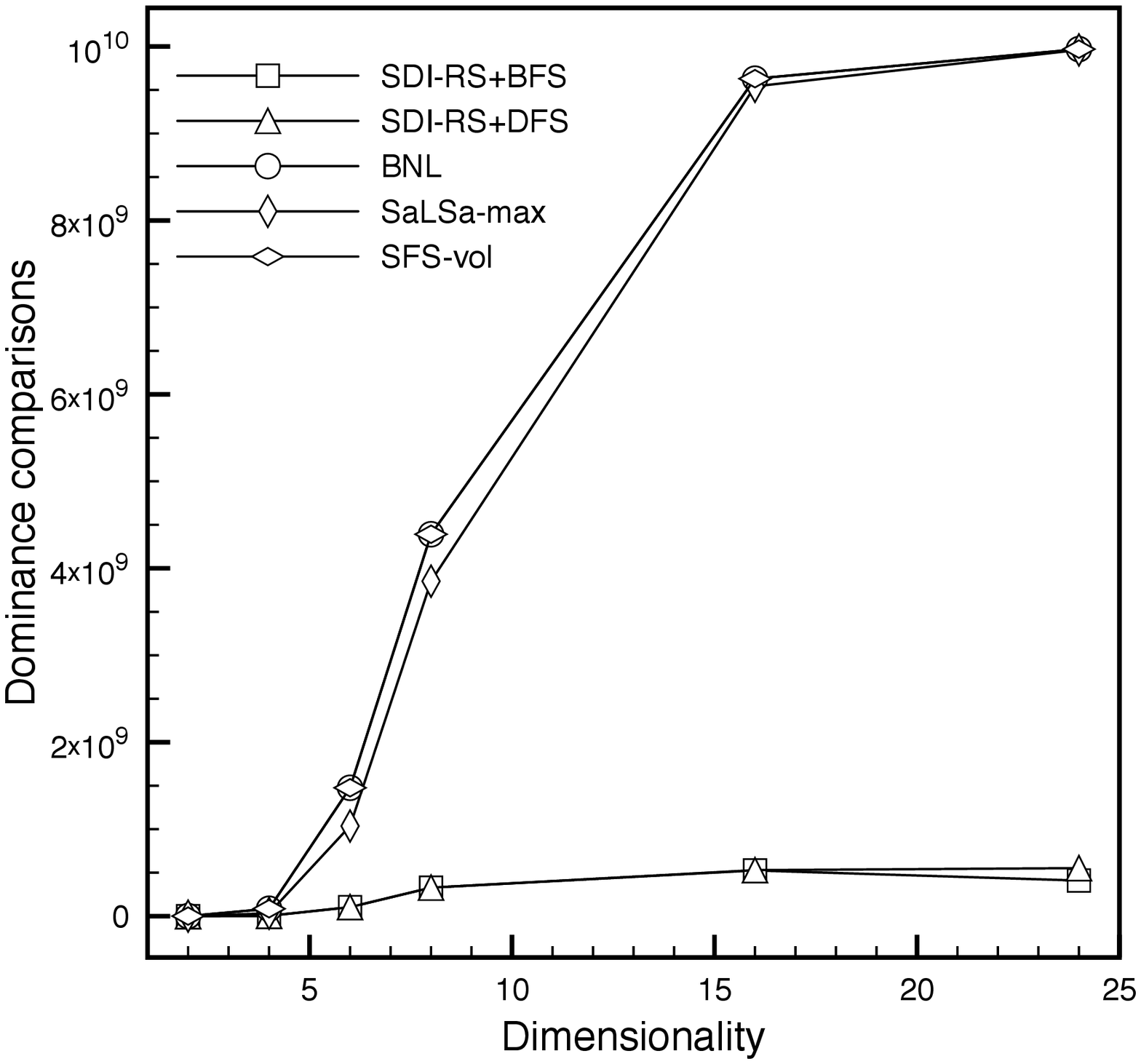}
&
\includegraphics[width=4cm]{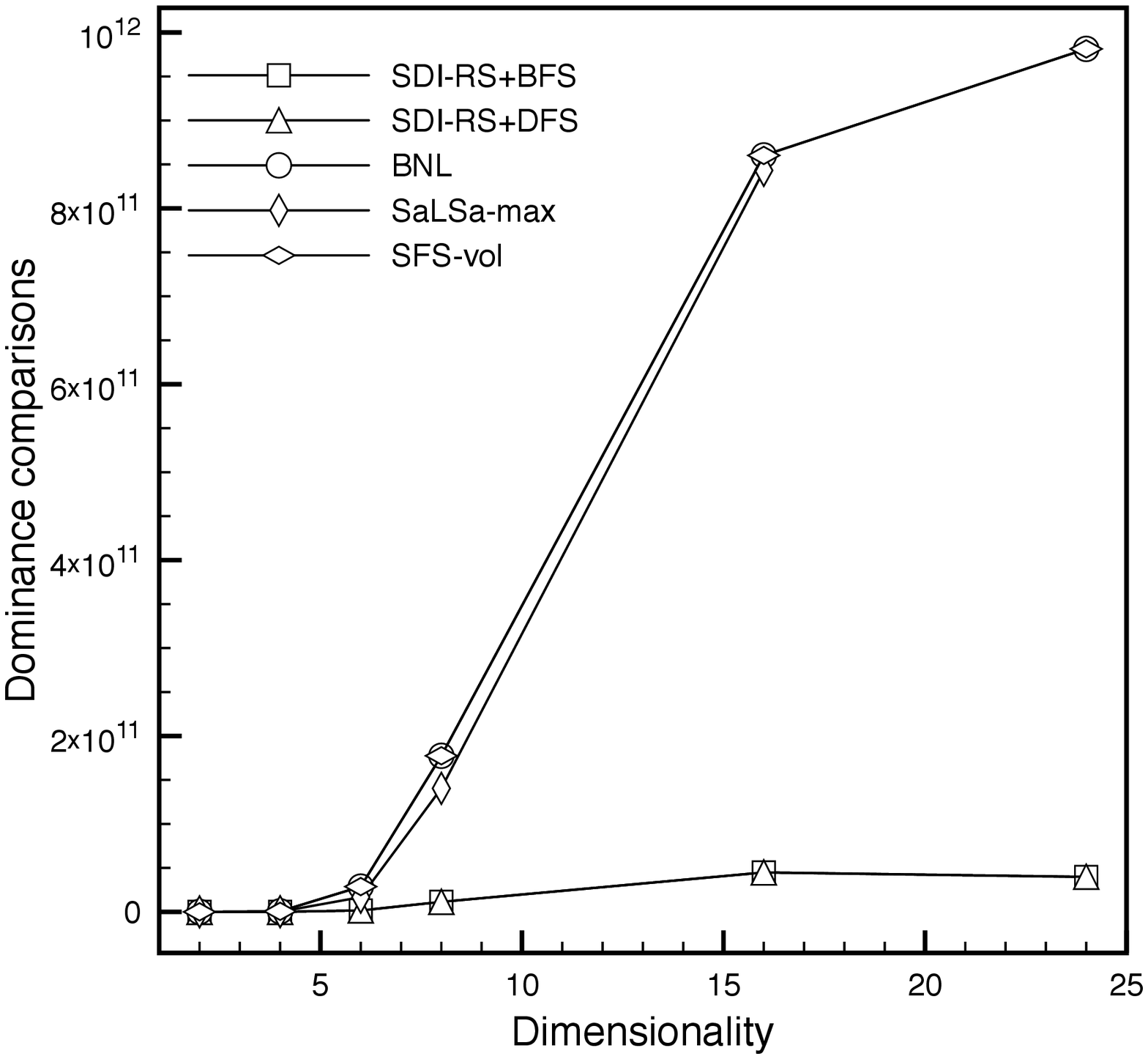}\\
(i) $n = 100$K & (j) $n = 1$M & (k) $n = 100$K & (l) $n = 1$M\\
\end{tabular}}
\end{center}
\caption{Overall performance of {\SDIRS}.}
\label{fig:perf}
\end{figure*}

We generate {\em independent}, {\em correlated}, and {\em anti-correlated} synthetic datasets using the standard Skyline Benchmark Data Generator\footnote{\url{http://pgfoundry.org/projects/randdataset}} \cite{Borzsony2001Operator} with the cardinality $n \in \{100K, 1M\}$ and the dimensionality $d$ in the range of 2 to 24.
Three real datasets {\tt NBA}, {\tt HOUSE}, and {\tt WEATHER} \cite{chester2015scalable} have also been used.
Table \ref{tab:syn} and Table \ref{tab:real} show statistics of all these datasets.

\begin{table}[htbp]
\begin{center}
{\scriptsize\begin{tabular}{r|l|rrrrrr}
\hline
\multicolumn{2}{c|}{Dataset} & $d = 2$ & $d = 4$ & $d = 6$ & $d = 8$ & $d = 16$ & $d = 24$\\
\hline
\hline
& 100K & 12 & 282 & 2534 & 9282 & 82546 & 99629\\
\cline{2-8}
Independent & 1M & 17 & 423 & 6617 & 30114 & 629091 & 981611\\
\hline
& 100K & 3 & 9 & 49 & 135 & 3670 & 13479\\
\cline{2-8}
Correlated & 1M & 1 & 19 & 36 & 208 & 8688 & 58669\\
\hline
& 100K & 56 & 3865 & 26785 & 55969 & 96816 & 99730\\
\cline{2-8}
Anti-correlated & 1M & 64 & 8044 & 99725 & 320138 & 892035 & 984314\\
\hline
\end{tabular}}
\end{center}
\caption{Skyline size of synthetic datasets.}
\label{tab:syn}
\end{table}

\begin{table}[htbp]
\begin{center}
{\scriptsize\begin{tabular}{l|rrr}
\hline
Dataset & Cardinality ($n$) & Dimensionality ($d$) & Skyline Size ($m$)\\
\hline
\hline
{\tt NBA} & 17264 & 8 & 1796\\
\hline
{\tt HOUSE} & 127931 & 6 & 5774\\
\hline
{\tt WEATHER} & 566268 & 15 & 63398\\
\hline
\end{tabular}}
\end{center}
\caption{Statistics of real datasets.}
\label{tab:real}
\end{table}

We implemented {\SDIRS} in C++ with {\tt C++11} standard, where dimensional indexes were implemented by STL {\tt std::vector} and {\tt std::sort()}.
In order to to evaluate the overall performance of our {\SDIRS}, three baseline algorithms {\BNL}, {\SFS}, and {\SALSA} were also implemented in C++ with the same code-base.
All algorithms are compiled using {\tt LLVM Clang} with {\tt -O3} optimization flag.
All experiments have been performed on a virtual computation node with 16 vCPU and 32GB RAM hosted in a server with 4 Intel Xeon E5-4610 v2 2.30GHz processors and 256GB RAM.

Figure \ref{fig:perf} shows the overall run-time, including loading/indexing data, and the total dominance comparison count of {\SDIRS} and {\BNL}/{\SFS}/{\SALSA} on 100K and 1M datasets, where the dimensionality is set to 2, 4 6, 8, 16, and 24.
We note that in the case of low-dimensional datasets, such as $d \leq 8$, there are no very big differences between all these 4 algorithms; however, {\SDIRS} extremely outperforms {\BNL}/{\SFS}/{\SALSA} in high-dimensional datasets, for instance $d \geq 16$.
Indeed, the run-time of {\SDIRS} is almost linear with respect to the increase of dimensionality, which is quite reasonable since the main cost in skyline computation is dominance comparison and {\SDIRS} allows to significantly reduce the total count of dominance comparisons.
On the other hand, it is surprising that {\SALSA} did not finish computing on all 24-dimensional datasets as Figure \ref{fig:perf} (b) and Figure \ref{fig:perf} (j), for more than 5 hours.
Notice that {\SALSA} outperforms {\BNL} and {\SFS} on real datasets.

We note that the total run-time of {\SDIRS} on low-dimensional correlated datasets is much than {\BNL}/{\SFS}/{\SALSA} as regards independent and anti-correlated datasets, because {\SDIRS} requires building dimensional indexes.
Table \ref{tab:time} details the skyline searching time $t_S$ (the time elapsed on dominance comparisons and data access in msec) and total run-time $t_T$ (the time elapsed on the whole process, including data loading and sorting/indexing in msec).
It is clear that the construction of dimensional indexes in {\SDIRS} is essential while the total processing time is short.

\begin{table}[htbp]
\begin{center}
{\scriptsize\begin{tabular}{r|l|r|r|r|r|r|r|r|r}
\hline
\multicolumn{2}{c|}{} & \multicolumn{2}{c|}{$d = 2$} & \multicolumn{2}{c|}{$d = 4$} & \multicolumn{2}{c|}{$d = 6$} & \multicolumn{2}{c}{$d = 8$}\\
\hline
\multicolumn{2}{c|}{} & $t_S$ & $t_T$ & $t_S$ & $t_T$ & $t_S$ & $t_T$ & $t_S$ & $t_T$\\
\hline
{\SDIRS} & 100K & 0.14 & 271 & 0.38 & 664 & 423 & 1657 & 243 & 1348\\
\cline{2-10}
 +{\BFS} & 1M & 0.17 & 5161 & 0.64 & 10293 & 1.34 & 23388 & 9896 & 40830\\
\hline
{\SDIRS} & 100K & 0.17 & 238 & 0.25 & 692 & 0.98 & 1263 & 4.75 & 1443\\
\cline{2-10}
+{\DFS} & 1M  & 0.14 & 5772 & 0.42 & 12069 & 1.32 & 21845 & 5.86 & 33012\\ 
\hline
{\BNL} & 100K & 2.51 & 75 & 5.36 & 141 & 5.33 & 243 & 10.29 & 284\\
\cline{2-10}
& 1M & 25.49 & 744 & 45.13 & 1468 & 62.56 & 2267 & 94.04 & 2901\\
\hline
{\SALSA} & 100K & 2686 & 2784 & 386 & 543 & 26.03 & 278 & 51.67 & 361\\
\cline{2-10}
{\tt max} & 1M & 88.63 & 1067 & 451 & 2117 & 377 & 2987 & 674 & 3829\\
\hline
{\SFS} & 100K & 1.49 & 91 & 4.53 & 162 & 4.91 & 263 & 10.99 & 320\\
\cline{2-10}
{\tt vol} & 1M & 13.33 & 931 & 33.78 & 1605 & 43.04 & 2346 & 88.83 & 3128\\
\hline
\end{tabular}}
\end{center}
\caption{Skyline searching time (msec) total run-time (msec) on correlated datasets.}
\label{tab:time}
\end{table}

Table \ref{tab:perf-real} shows the performance of {\SDIRS} on real datasets.
{\DFS} dimension switching outperforms {\BFS} dimension switching on both {\tt NBA} and {\tt HOUSE} datasets however {\BFS} outperforms {\DFS} on {\tt WEATHER} dataset.
After having investigated these datasets, we confirm that there are a large number of duplicate values in several dimension of {\tt WEATHER} dataset so the {\BFS} dimension switching strategy takes its advantage.
{\BNL} outperforms all other tested algorithms on {\tt HOUSE} dataset, which corresponds to the results obtained from synthetic low-dimensional independent datasets.
Furthermore, the update numbers of the best stop line in {\SDIRS} is quite limited with respect to the size of skylines.

\begin{table}[htbp]
\begin{center}
{\scriptsize\begin{tabular}{l|r|r|r|r|r}
\hline
 & {\SDIRS}+{\BFS} & {\SDIRS}+{\DFS} & {\BNL} & {\SALSA} & {\SFS}\\
\hline
\hline
Dominance & 680,388 & 662,832 & 8,989,690 & 6,592,178 & 8,989,690\\
\hline
Search Time (msec) & 54 & 38 & 151 & 108 & 147\\
\hline
Total Time (msec) & 172 & 158 & 191 & 152 & 189\\
\hline
Stop Line Update & 15 & 32 & -- & -- & --\\
\hline
\end{tabular}}\\
(a) {\tt NBA} dataset: $d = 8$, $n = 17264$, $m = 1796$.\\
~\\
{\scriptsize\begin{tabular}{l|r|r|r|r|r}
\hline
 & {\SDIRS}+{\BFS} & {\SDIRS}+{\DFS} & {\BNL} & {\SALSA} & {\SFS}\\
\hline
\hline
Dominance & 4,976,773 & 4,860,060 & 59,386,118 & 51,484,870 & 59,386,118\\
\hline
Search Time (msec) & 962 & 337 & 1,486 & 1,550 & 1,534\\
\hline
Total Time (msec) & 2,663 & 1,918 & 1,716 & 1,800 & 1,768\\
\hline
Stop Line Update & 16 & 18 & -- & -- & --\\
\hline
\end{tabular}}\\
(b) {\tt HOUSE} dataset: $d = 6$, $n = 127931$, $m = 5774$.\\
~\\
{\scriptsize\begin{tabular}{l|r|r}
\hline
 & {\SDIRS}+{\BFS} & {\SDIRS}+{\DFS}\\
\hline
\hline
Dominance & 1,744,428,382 & 1,737,143,260\\
\hline
Search Time (msec) & 48,773 & 58,047\\
\hline
Total Time (msec) & 65,376 & 77,665\\
\hline
Stop Line Update & 14 & 18\\
\hline
\end{tabular}}\\
{\scriptsize\begin{tabular}{l|r|r|r}
\hline
 & {\BNL} & {\SALSA} & {\SFS}\\
\hline
\hline
Dominance & 14,076,080,681 & 7,919,746,895 & 14,076,080,681\\
\hline
Search Time (msec) & 539,100 & 394,995 & 545,263\\
\hline
Total Time (msec) & 541,820 & 397,650 & 547,914\\
\hline
Stop Line Update & -- & -- & --\\
\hline
\end{tabular}}\\
(c) {\tt WEATHER} dataset: $d = 15$, $n = 566268$, $m = 63398$.\\
~\\
\end{center}
\caption{Performance evaluation on real datasets.}
\label{tab:perf-real}
\end{table}

We did not directly compare {\SDIRS} with all existing skyline algorithms, but with reference to most literature comparing proposed algorithms with {\BNL}, {\SFS}, or {\SALSA}, the comparative results obtained in our experimental evaluation indicate that {\SDIRS} outperforms the most of existing skyline algorithms.

\section{Conclusion}

In this paper, we present a novel efficient skyline computation approach.
We proved that in multidimensional databases, skyline computation can be conducted on an arbitrary dimensional index which is constructed with respect to a predefined total order that determines the skyline, we therefore proposed a dimension indexing based general skyline computation framework {\SDI}.
We further showed that any skyline tuple can be used to stop the computation process by outputting the complete skyline.
Based on our analysis, we developed a new progressive skyline algorithm {\SDIRS} that first builds sorted dimensional indexes then efficiently finds skyline tuples by dimension switching in order to minimize the count of dominance comparisons.
Our experimental evaluation shows that {\SDIRS} outperforms the most of existing skyline algorithms.
Our future research direction includes the further development of the {\SDI} framework as well as adapting the {\SDI} framework to the context of Big Data, for instance with the Map-Reduce programming model.

\bibliography{article}
\bibliographystyle{abbrv}

\end{document}